\documentclass[11pt]{article}

\usepackage[T1]{fontenc}
\usepackage{lmodern}
\usepackage{natbib}       
\usepackage{amsmath}                        
\numberwithin{equation}{section}
\usepackage{graphicx} 											
\usepackage[normalem]{ulem}
\usepackage{epstopdf}	
\usepackage{enumitem}										
\usepackage{mdwlist}											
\usepackage[dvipsnames]{xcolor}					
\usepackage[plainpages=false, pdfpagelabels]{hyperref} 
	\hypersetup{
		colorlinks   = true,
		citecolor    = NavyBlue, 
		linkcolor    = RubineRed, 
		urlcolor     = Turquoise
	}
\usepackage{amssymb}                        
\usepackage[mathscr]{eucal}                 
\usepackage{dsfont}													
\usepackage[paperwidth=8.5in,paperheight=11in,top=1.25in, bottom=1.25in, left=1.0in, right=1.0in]{geometry} 
\usepackage{mathtools}                      
\mathtoolsset{showonlyrefs=true}          

\usepackage{setspace}
\usepackage{amsthm}                         
	\allowdisplaybreaks                       
	\theoremstyle{plain}
	\newtheorem{theorem}{Theorem}
	\numberwithin{theorem}{section}
	\newtheorem{lemma}[theorem]{Lemma}       	
	\newtheorem{proposition}[theorem]{Proposition}
	\newtheorem{corollary}[theorem]{Corollary}
	\theoremstyle{definition}
	
	\newtheorem{example}[theorem]{Example}
	\newtheorem{remark}[theorem]{Remark}
	\newtheorem{assumption}[theorem]{Assumption}
	\newtheorem*{coro-non}{Corollary}

\renewcommand{\[}{\left[}

\DeclareMathOperator*{\argmax}{arg\,max}

\newcommand\Eb{\mathds{E}}
\newcommand\Fb{\mathds{F}}
\newcommand\Ib{\mathds{1}}

\newcommand\Pb{\mathds{P}}

\newcommand\Rb{\mathds{R}}

\newcommand\Nb{\mathds{N}_n}

\newcommand\Ac{\mathds{A}}
\newcommand\Acc{\mathscr{A}}

\newcommand\Dc{\mathscr{D}}

\newcommand\Fc{\mathscr{F}}

\newcommand\Lc{\mathscr{L}}
\newcommand\Mc{\mathscr{M}}

\newcommand\eps{\varepsilon}
\newcommand\om{\omega}
\newcommand\Om{\Omega}
\newcommand\sig{\sigma}

\newcommand\Lam{\Lambda}
\newcommand\gam{\gamma}

\newcommand\lam{\lambda}

\newcommand\Del{\Delta}

\renewcommand\d{\partial}

\newcommand\dd{\mathrm{d}}
\newcommand\ee{\mathrm{e}}

\newcommand{\e}{{\bf e}}
\renewcommand{\j}{{\bf j}}

\newcommand{\N}{\mathbb{N}}

\newcommand{\Et}{\mathds{E}_{t,x,p,q}}

\usepackage{soul}
\begin{document}

\title{Optimal Bookmaking}

\author{
Matthew Lorig\thanks{Department of Applied Mathematics, University of Washington, Seattle, WA, USA.  \textbf{e-mail}: \url{mlorig@uw.edu}.}
\and
Zhou Zhou\thanks{School of Mathematics and Statistics, University of Sydney, Sydney, Australia.  \textbf{e-mail}: \url{zhou.zhou@sydney.edu.au}.}
\and
Bin Zou\thanks{Corresponding author. 341 Mansfield Road U1009, Department of Mathematics, University of Connecticut, Storrs, CT 06269-1009, USA. \textbf{e-mail}: \url{bin.zou@uconn.edu}. Phone: +1-860-486-3921.}
}

\date{\vspace{1cm} Accepted for publication in\\ 
	\emph{European Journal of Operational Research}\\First Version: July 2019\\This Version: March 2021}

\maketitle


\begin{abstract}
We introduce a general framework for continuous-time betting markets, in which a bookmaker can dynamically control the prices of bets on outcomes of random events.
In turn, the prices set by the bookmaker affect the rate or intensity of bets placed by gamblers.  The bookmaker seeks an optimal price process that maximizes his expected (utility of) terminal wealth. We obtain explicit solutions or characterizations to the bookmaker's optimal bookmaking problem in various interesting models.
\end{abstract}

\noindent
\textbf{Keywords:} 
Stochastic programming;
Poisson process; 
sports betting; 
stochastic control; 
utility maximization

\vspace{2ex}

\noindent
\textbf{Declaration of interest:} none.

\newpage
%
%

\section{Introduction}
\label{sec:intro}
Sports betting is a large and fast-growing industry.  
According to a recent report by \cite{zion2018}, the global sports betting market was valued at around 104.31 billion USD in 2017 and is estimated to reach approximately 155.49 billion by 2024.  As noted by \cite{sbj2018}, growth within the United States is expected to be particularly strong due to a May 2018 ruling by the Supreme Court, which deemed the Professional and Amateur Sports Protection Act (PASPA) unconstitutional and thereby paved the way for states to rule  individually on the legality of sports gambling.  Less than a year after the PASPA ruling, eleven states had passed bills legalizing gambling on sports compared to just one (Nevada) prior to the PASPA ruling.  At the time of the writing of this paper, several additional states have drafted bills to legalize sports betting within their borders.\footnote{Please see update on the American Gaming Association \url{https://www.americangaming.org/research/state-gaming-map/}.}

In this paper, we analyze sports betting markets from the perspective of a \textit{bookmaker}.  A bookmaker sets prices for bets placed on different outcomes of sporting events (or random events in general), collects revenue from these bets, and pays out winning bets.  Often the outcomes are binary (team A wins or team B wins).  But, some events may have more than two outcomes (e.g., in an association football match, a team may win, lose or tie).  Moreover, the outcomes need not be mutually exclusive (e.g., separate bets could be placed on team A winning and player X scoring).  Ideally, a bookmaker strives to set prices in such a way that, no matter what the outcome of a particular event, he collects sufficient betting revenue to pay out all winning bets while also retaining some revenue as profit.  However, because a bookmaker only controls the prices of bets -- \textit{not} the number of bets placed on different outcomes -- he may lose money on a sporting event if a particular outcome occurs.  

Consider, for example, a sporting event with only two mutually exclusive outcomes: A and B.  Suppose the bookmaker has received many more bets on outcome A than on outcome B.  If outcome A occurs, then the bookmaker would lose money unless the revenue he has collected from bets placed on outcome B is sufficient to cover the payout of bets placed on outcome A.  Thus, if a bookmaker finds himself in a situation in which he has collected many more bets on outcome A than on outcome B, he may raise the price of a bet placed on outcome A and/or lower the price of a bet placed on outcome B.  This strategy would lower the demand for bets placed on outcome A and increase the demand for bets placed on outcome B.  As a result, the bookmakers would reduce the risk that he incurs a large loss, but he would do so at the cost of sacrificing expected profits.

Traditionally, bookmakers would only take bets prior to the start of a sporting event.  In such cases, the probabilities of particular outcomes would remain fairly static as the bookmaker received bets.  More recently, however, bookmakers have begun to take bets on sporting events as the events occur.  In these cases, the probabilities of particular outcomes evolve stochastically in time as the bookmaker takes bets.  This further complicates the task of a bookmaker who, in addition to considering the number of bets he has collected on particular outcomes, must also consider the dynamics of the sporting event in play.  For example, scoring a point near the end of a basketball game when the score is tied would have a much larger effect on the outcome of the game than scoring a point in the first quarter.  And, 
the effect of scoring a goal in the first half of an association football match would be greater than the effect of scoring a goal early in a game of basketball.

In this paper, we provide a \textit{general framework} for studying optimal bookmaking, which takes into account the situations described above.
In our framework, the probabilities of outcomes are allowed to either be static or to evolve stochastically in time,
and bets on any particular outcomes may arrive either at a deterministic rate or at a stochastic intensity.
This rate or intensity is a decreasing function of the price the bookmaker sets and an increasing function of the probability that the outcome occurs.
In this setting a bookmaker may seek to maximize either expected wealth or expected utility from wealth. 
We provide several examples of in-game dynamics.
And we analyze specific optimization problems in which the bookmaker's optimal strategy can be obtained explicitly.

Existing literature on optimal bookmaking in the context of stochastic control appears to be rather scarce.
In one of the few existing papers on the topic, \cite{hodges2013fixed} consider the bookmaker of a horse race.
In their paper, the probability that a given horse wins the race is fixed and the number of bets placed on a given horse is a normally distributed random variable with a mean that is proportional to the probability that the horse will win and inversely proportional to the price set by the bookmaker.
In this one-period setting, the bookmaker seeks to set prices in order to maximize his utility from terminal wealth.
Unlike our paper, \cite{hodges2013fixed} do not allow the bookmaker to dynamically adjust prices as bets come in nor do they find an analytic solution for the bookmaker's optimal control.  As far as we are aware, our paper is the first to provide a general framework for studying optimal bookmaking in a dynamic setting.

Similar to our work, \cite{divos2018} consider dynamic betting during an association football match.  However, rather than approach bookmaking as an optimal control problem, they use a no-arbitrage replication argument to determine the value of a bet whose payoff is a function of the goals scored by each of the two teams.  As ``hedging assets'' they consider bets that pay the final goal tally of the two teams.  By contrast, in our work, the bookmaker has no underlying assets to use as hedging instruments. 
Arbitrage opportunities in soccer betting markets are also studied by \cite{grant2018new}.

As mentioned above, we study a stochastic control problem for a bookmaker who seeks to set prices of bets optimally in order to maximize his expected profit or utility. The sports betting system in our framework is fundamentally different from a parimutuel betting system, in which all winning bets share the funds in the pool. \cite{bayraktar2017high} study the impacts of large bettors on the house and small bettors in parimutuel wagering events. Note that \cite{bayraktar2017high} model parimutuel wagering as a static game and assume the event has two mutually exclusive results. We refer to \cite{thaler1988anomalies},  \cite{hausch2011handbook}, and the references therein for further discussions on parimutuel betting.

We comment that our research of optimal bookmaking, in terms of motivation, problems, and methodology, is fundamentally different from those statistical modeling works in sports betting. 
On the statistical modeling side, most existing papers attempt to model and then forecast 
the performance of a player (or a team) and  the outcomes of random (sporting) events, and analyze the impact of various factors on individual or team performance in sports. 
Some works further apply the fitted model to construct a betting strategy to exploit opportunities, mostly from the perspective of a bettor. 
For instance, in an early work, \cite{dixon1997modelling} use a regression model of Poisson distributions to fit the football (soccer) data in England and obtain a betting strategy that delivers a positive return. 
\cite{klaassen2003forecasting} propose a statistical method to forecast the winner of a tennis game.
\cite{bozoki2016application} apply pairwise comparison matrices to rank top tennis players.
A recent paper by \cite{song2020gamma} applies a gamma process to model the dynamic in-game scoring of a National Basketball Association (NBA) game. 
In other related areas, \cite{muller2017beyond} apply multilevel regression analysis to estimate the market value of professional soccer players in top European leagues.
We refer to the survey article \cite{wright2014or} and the references therein for a comprehensive review of the application of operations research and analytics in the laws and rules of various sports.

In certain respects, the optimal bookmaking problem we consider is similar to the optimal market making problem considered by
\cite{avellaneda2008high, gueant2012optimal, adrian2019intraday} and the optimal execution problem analyzed in \cite{gatheral2013dynamical, bayraktar2014liquidation, cartea2015}.  In these papers, a market maker offers limit orders to buy and sell a risky asset whose reference price is a stochastic process.  The intensity at which limit orders are filled is a decreasing function of how far below (above) the reference price the limit order to buy (sell) is.  In this setting, a market maker seeks to maximize either his expected wealth or his expected utility from wealth, which includes both cash generated from filled limit orders and the value of any holdings in the risky asset.  In some cases, the market maker also seeks to minimize his holdings in the risky asset.  In a sense, a market maker offering limit orders to buy and sell a risky asset is akin to a bookmaker taking bets on mutually exclusive outcomes during a sporting event.  And, a market maker seeking to minimize holdings in a risky asset is similar to a bookmaker seeking to have equal money bet on mutually exclusive outcomes. 
A recent paper \cite{capponi2019large} unifies the above mentioned two strands of literature. 
They consider a large uninformed seller who wants to optimally execute an order, accounting for the response of a group of competitive market makers, under a continuous-time Stackelberg game framework. 
The large seller determines his sell intensity in order to maximize his expected net proceeds, while market markers decide their optimal buy/sell amount to maximize their expected gross wealth minus a quadratic cost of running inventory. 
On the modeling side, \cite{capponi2019large} model the clearing bid and ask prices as a function of the market makers' controls (amount offered), which then affect the end-users' buy/sell intensity and the large seller's sell intensity. 
Similarly, we model the betting intensity as a function of the bookmaker's control (price).
Both the market makers in \cite{capponi2019large} and the bookmaker in our paper want to maximize their  expected (utility of) wealth, but we do not consider the competition among different bookmakers. 

We remark that the optimal bookmaking problem considered in this paper is also related to stochastic control problems with jump-diffusion dynamics from a methodological point of view, as both attempt to control the ``jump events'' in an optimal way.
\cite{zou2014optimal} study an optimal investment and risk control problem for an insurer whose risk process is modeled by a jump-diffusion process.
\cite{capponi2014dynamic} study optimal investment in a defaultable market with regime switching, where an additional defaultable bond with recovery is available for investment.
\cite{bo2016optimal} consider a power utility maximizing agent who invests in several credit default swaps (CDSs) and a money market, where the default intensity is piece-wise constant between consecutive jumps to defaults.
\cite{bo2019credit} further extend the setup in \cite{bo2016optimal} by applying a diffusion process to model the default intensity between consecutive default events. 
Apart from the obvious different context of problem formulation, there exists a significant difference between this paper and the above-mentioned papers: the jump intensity  in this paper is endogenous and depends on the control (prices of bets), but is exogenous and cannot be controlled in the above-mentioned papers.

The rest of this paper proceeds as follows.
In Section \ref{sec:general}, we present a general framework for continuous-time betting markets, state the bookmaker's optimization problem and define his value function.
In Section \ref{sec:PDE}, we provide a characterization of the bookmaker's value function as the solution of a partial (integro)-differential equation (PDE). 
In Section \ref{sec:cont}, we study the bookmaker's optimization problem in a semi-static setting.  In this section, the probabilities of outcomes remain constant in time, but bets arrive at rates that depend on the prices set by the bookmaker.
In Section \ref{sec:wealth_max}, we study the wealth maximization problem for a risk-neutral bookmaker.
In Section \ref{sec:exp}, we focus on the bookmaker's optimization problem when his preferences are characterized by a utility function of exponential form.
Lastly, in Section \ref{sec:con}, we offer some concluding remarks. 
Several technical proofs and results are collected in an online appendix.

%
%

\section{A General Framework for Continuous Time Betting Markets}
\label{sec:general}

Let us fix a probability space $(\Om,\Fc,\Pb)$ and a filtration $\Fb = (\Fc_t)_{0 \leq t \leq T}$ completed by $\Pb$-null sets, where $T < \infty$ is a finite time horizon. 
We shall suppose that all the stochastic processes and random variables introduced in this paper are well-defined and adapted under the given filtered probability space, which may be further characterized for specific examples or models in the sequel.
We will think of $\Pb$ as the \textit{real world} or \textit{physical} probability measure.
Consider a finite number of subsets $(A_i)_{i \in \Nb}$ (where $\Nb:=\{1,2,\cdots,n\}$) of $\Om$, each of which is $\Fc_T$-measurable (i.e., $A_i \in \Fc_T$ for all $i$). 
We can think of $A_i$ as a particular set of outcomes of a sporting event,
which finishes at time $T$. 
To avoid trivial cases, we suppose $\Pb(A_i) \in (0,1)$ for all $i$.
Note that the sets $(A_i)_{i \in \Nb}$ need not be a partition of $\Om$.  In particular, there may be outcomes $\om$ such that $\om \notin \cup_{i=1}^n A_i$.
Moreover, sets may overlap (i.e., we may have $A_i \cap A_j \neq \emptyset$ where $i \neq j$).

We will denote by $P^i = ( P_t^i )_{0 \leq t \leq T}$ the $\Fc_t$-conditional probability of $A_i$.  We have 
\begin{align}
\label{eq:prob_def}
P_t^i	&=	\Eb_t  \Ib_{A_i}  ,
\end{align}
where $\Eb_t \, \cdot \, := \Eb( \, \cdot \, | \Fc_t )$ denotes conditional expectation.  Note that, by the tower property of conditional expectation, $P^i$ is a $(\Pb,\Fb)$-martingale.  We will denote by $P = (P^1, P^2, \cdots, P^n)$ the vector of conditional probabilities. 
In general, the conditional probability $P^i$ is a stochastic process.
However, we will also consider scenarios where it is reasonable to assume that $P_t^i$ is a fixed constant $p_i$ for all $t < T$.
Let us take a look at some examples of sporting events and show how we can describe them probabilistically.

\begin{example}
\label{ex:n-goals}
Suppose the number of goals scored by player X in an association football match is a Poisson process with intensity $\mu$.
Consider an in-game bet on a set of outcomes of the form $A_i =$ ``player X will finish the match with exactly $i$ goals.''
We can model the conditional probability that $A_i$ occurs as follows.
Fix a probability space $(\Om,\Fc,\Pb)$ and let $\Fb = (\Fc_t)_{0 \leq t \leq T}$ be the augmented filtration generated by the Poisson process $N^\mu = (N_t^\mu)_{0 \leq t \leq T}$ with intensity $\mu$.  Then we have 
$P_t^i = \sum_{n=0}^i \Pb( N_T^\mu - N_t^\mu = i - n ) \Ib_{\{N_t^\mu = n\}} = \sum_{n=0}^i \frac{1}{(i-n)!} \ee^{-\mu(T-t)}[\mu(T-t)]^{i-n} \,  \Ib_{\{N_t^\mu = n\}} $, 
from which we can easily construct the conditional probabilities of outcomes of the form ``player X will finish the match with between $i$ and $j$ goals (inclusive).''
For example, if $A_j =$ ``player X will score at least one goal'' then we have 
$P_t^j = \Ib_{\{N_t^\mu \geq 1 \}} + \Ib_{\{N_t^\mu = 0\}} ( 1 - \ee^{-\mu(T-t)} )$.
The dynamics of $P^j$ can be easily deduced as 
$\dd P_t^j=	\Ib_{\{P_{t-}^j < 1\}} \ee^{-\mu(T-t)}( \dd N_t^\mu - \mu \dd t)$, with $P_0^j
=	1 - \ee^{- \mu T} $ .
Observe that the above $P^j$  is a martingale, as it must be.
\end{example}

\begin{example}
\label{ex:n-points}
Consider an National Basketball Association (NBA) game.
Although points in NBA games are integer-valued, as the number of points in an NBA game is on the order of 100, it is reasonable to approximate the number of points scored as a $\Rb$-valued process.  To this end, let us consider a probability space $(\Om, \Fc, \Pb)$ equipped with a filtration $\Fb = (\Fc_t)_{0 \leq t \leq T}$ for a Brownian motion $W = (W_t)_{0 \leq t \leq T}$.  We can model the point differential $\Del = (\Del_t)_{0 \leq t \leq T}$ between team A and team B as a Brownian motion with drift $\Del_t = \mu t + \sig W_t$ where the sign and size of $\mu$ captures how much team A is favored by.   
Now, consider a bet on a set of outcomes of the form $A_i =$ ``team A will win by $i$ points or more.''  Then we have 
\begin{align}
P_t^i
	&=	\Pb( \Del_T \geq i | \Del_t)
	=	1 - \Phi \Big( \frac{i - \Del_t - \mu (T-t) }{ \sig \sqrt{T-t}} \Big)
	=		\Phi \Big( \frac{ \Del_t+  \mu (T-t) - i }{ \sig \sqrt{T-t}} \Big) , 
	\label{eq:prob_diff}
\end{align}
where $\Phi$ is the cumulative distribution function (c.d.f.) of a standard normal random variable.
We can easily obtain the dynamics of $P_t^i$ using It\^o's Lemma by
$\dd P_t^i=	\Phi' (  \Phi^{-1}(P_t^i) ) / \sqrt{T-t} \,  \dd W_t $.
Observe that $P^i$ is a martingale, as it must be.
\end{example}

Having seen a few examples of how we can describe sporting events probabilistically, let us now focus on the payoff structure of bets.
Throughout this paper, we will assume that
a bet placed on a set of outcomes $A_i$ pays one unit of currency at time $T$ if and only if $\om \in A_i$ (i.e., if $A_i$ occurs).  Thus, we have
\begin{align}
\label{eq:payoff}
\text{``payoff of a bet placed on $A_i$''}
	&=	\Ib_{A_i} = P_T^i .  
\end{align}

\begin{remark}
In the US, if a bookmaker quotes odds of +120 on outcome A, then a \$100 bet on outcome A would pay $\$120 + \$100 = \$220$ if A occurs.  If a bookmaker quotes odds of -110 on outcome B, then a \$110 bet on outcome B would pay $\$100 + \$110 = \$210$ if B occurs. 
Using our setup,  the price of a bet quoted at $+120$ is $100/220 \approx 0.4545$ and the price of a  bet quoted at $-110$  is $110/210\approx 0.5238$. In general, a quote of $+x$ ($x>100$) is equivalent to setting the price to $\frac{100}{x+100} < 0.5$ (underdogs in sports betting) and a quote of $-x$ ($x>100$) is equivalent to setting the price to $\frac{x}{x+100} > 0.5$ (favorites in sports betting).
The payoff in \eqref{eq:payoff} is simply a convenient normalization, which can be applied to any bet that pays a fixed (i.e., non-random) amount if a particular outcome or set of outcomes occurs. 
\end{remark}

In our framework, the bookmaker cannot control the number of bets that the public places on a set of outcomes $A_i$ directly.  However, the bookmaker can set the price of a bet placed on $A_i$ and this will affect the rate or intensity at which bets on $A_i$ are placed.  
Such a setup (i.e., controlling the price to influence demand) is similar to the setup used in operations research literature to study the problem of selling seasonal and style goods in various industries;
see, e.g., \cite{pashigian1988demand} and \cite{gallego1994optimal}.
We will denote by $u^i = (u_t^i)_{0 \leq t < T }$ the price set by the bookmaker of a bet placed on $A_i$.  The vector of prices will be denoted as $u = (u^1, u^2, \cdots, u^n)$.  
It will be helpful at this stage to introduce the set of \textit{admissible pricing strategies} $\Acc(t,T)$, which we define as follows
\begin{align}
\Acc(t, T) 
	:= \{ u = (u_s)_{s \in [t,T)} : u \text{ is progressively measure w.r.t. } \Fb \text{ and }u_s \in \Ac:=		[0,1]^n \} , \label{eq:Au} 
\end{align}
with 
$t \in [0,T)$. 
Note that we do not include the control $u_T$ in the definition of $\Acc(t,T)$, as the case at time $T$ is trivial. 
Without loss of generality, we set $u^i_T = \Ib_{A_i}$ for all $i \in \Nb$.

Let us denote by $X^{u} = (X_t^u)_{0 \leq t \leq T}$ the total revenue generated by the bookmaker and 
by $Q^{u,i} = (Q_t^{u,i})_{0 \leq t \leq T}$ the total number of bets placed on a set of outcomes $A_i$. 
Note that we have indicated with a superscript the dependencies of $X^u$ and $Q^u$ on bookmaker's pricing policy $u$. 
The relationship between $X^u$, $Q^u$ and $u$ is 
$\dd X_t^u = \sum_{i =1}^n u_t^i \, \dd Q_t^{u,i}. $
Observe that $X^u$ and $Q^{u,i}$ for all $i \in \Nb$ are non-decreasing processes.
Typically, we will have $X_0^u \equiv X_0 = 0$ and $Q_0^{u,i} \equiv Q_0^i=0$ for all $i \in \Nb$.  However, we do not require this (to account for the fact that the bookmaker may have taken the bets before time 0).

In this paper, we consider two models for $Q^{u,i}$.  In one model, bets on a set of outcomes $A_i$ arrive at a rate per unit time, which is a function $\lam_i : \Ac \times \Ac \to \bar\Rb_+$ of the vectors of conditional probabilities $P$ of outcomes and prices $u$ set by the bookmaker. Here, $\Ac=[0,1]^n $ and $\bar \Rb_+ = \Rb_+ \cup \{+\infty\}$.  Under this model, we have 
\begin{align}
\label{eq:det_arrival}
\text{Continuous arrivals}:&&
Q_t^{u,i} 
	&= \int_0^t \lam_i ( P_s , u_s ) \dd s + Q_0^{i} . 
\end{align}
In another model, bets on a set of outcomes $A_i$ arrive as a state-dependent Poisson process $N^{u,i} = (N^{u,i}_t)_{0 \leq t \leq T}$ whose instantaneous  intensity is a function $\lam_i : \Ac \times \Ac \to \bar\Rb_+$ of the vectors of conditional probabilities $P$ of outcomes and prices $u$ set by the bookmaker.  Under this model, we have 
\begin{align}
\label{eq:Poi_arrival}
\text{Poisson arrivals}:&&
Q^{u,i}_t 
	&= \int_0^t \dd N^{u,i}_t + Q_0^{i}, & 
\Eb_t  \dd N^{u,i}_t  
	&=	\lam_i ( P_
	t , u_t ) \dd t . 
\end{align}
Throughout the paper, we will refer to the function $\lam_i$ as the \textit{rate} function when bets arrive continuously
and the \textit{intensity} function when bets arrive as a state-dependent Poisson process. 
To summarize, we model the bet arrivals $Q^{u,i}$ as a controlled deterministic process in \eqref{eq:det_arrival} or a controlled Poisson process in \eqref{eq:Poi_arrival}. Such a modeling is similar to that of demand process in the production-pricing literature; see \cite{li1988stochastic} and \cite{gallego1994optimal} for details, and the survey article \cite{elmaghraby2003dynamic} and the references therein for general discussions.

\begin{remark}
For sporting events with a large betting interest (e.g., the Super Bowl, the UEFA Champions League final, etc.), the continuous arrivals model given by \eqref{eq:det_arrival} is sufficient to capture the dynamics of bet arrivals.  However, for sporting events with limited betting interest (e.g., curling in the Winter Olympics, the Westminster Dog Show, etc.), the dynamics of bet arrivals are better captured by the Poisson arrivals model in \eqref{eq:Poi_arrival}. 
\end{remark}

Although our framework is sufficiently general to allow for the rate/intensity function $\lam_i$ to depend on the entire vector of conditional probabilities $P$ and vector of prices $u$, the case in which $\lam_i$ depends only on the conditional probability $P^i$ and price $u^i$ will be of special interest. 
In cases for which $\lam_i ( P_t, u_t )  \equiv \lam_i ( P_t^i, u_t^i )$, where $i \in \Nb $ and $t \in [0,T]$, 
we expect $\lam_i$ to be an increasing function of the conditional probability $P^i$ of outcome $A_i$ and a decreasing function of the price $u^i$ set by the bookmaker. 
Examples of rate/intensity functions $\lam_i : [0,1] \times [0,1] \to \bar\Rb_+$ that satisfy those analytical properties include
\begin{align}
\lam_i(p_i,u_i) 
&:=	\frac{p_i}{1-p_i}  \frac{1 - u_i }{u_i} , \label{eq:lambda-1} \\
\lam_i(p_i,u_i) 
&:= \frac{ \log u_i }{\log p_i} . \label{eq:lambda-2}
\end{align}
The functions \eqref{eq:lambda-1} and \eqref{eq:lambda-2} have reasonable qualitative behavior because
(i) as the price $u_t^i$ of a bet on outcome $A_i$ goes to zero, the intensity of bets goes to infinity
\begin{align}
\lim_{u_i \to 0 } \lam_i(p_i,u_i)
&=	\infty , \label{limit-1}
\end{align}
(ii) as the price $u_t^i$ of a bet on outcome $A_i$ goes to one, the intensity of bets goes to zero
\begin{align}
\lim_{u_i \to 1 } \lam_i(p_i,u_i)
&=	0 ,  \label{limit-2}
\end{align}
and (iii) all \textit{fair bets} $u_t^i = P_t^i$ have the same intensity $\lam_i(p_i,p_i)=	\lam_i(q_i,q_i) .$
Another rate/intensity function $\lam_i$ we shall consider in this paper is
\begin{align}
\lam_i(p_i,u_i)
	&:= \kappa \ee^{-\beta (u_i - p_i)} , &
\kappa, \beta
	&> 0 . \label{eq:lambda-3}
\end{align}
As we shall see, the form of $\lam_i$ in \eqref{eq:lambda-3} facilitates analytic computation of optimal pricing strategies. 
The exponential form \eqref{eq:lambda-3} of the intensity function is also used in \cite{gallego1994optimal}[Section 2.3] to obtain optimal pricing strategies in closed-form.
Note, however, that $\lam_i$ in \eqref{eq:lambda-3} does not satisfy \eqref{limit-1} or \eqref{limit-2}. 
We mention that we can also multiple the intensity functions in \eqref{eq:lambda-1} and \eqref{eq:lambda-2} by a positive scaling factor $\kappa$ as in \eqref{eq:lambda-3}.

Let us denote by $Y_T^u$ the total wealth of the bookmaker just \textit{after} paying out all winning bets, assuming he follows pricing policy $u$.  Then we have
\begin{align}
Y_T^u
	&=	 X_T^u - \sum_{i = 1}^n P_T^i Q_T^{u,i} 
	= 	X_0^u - \sum_{i=1}^n P_T^i Q_0^i + \sum_{i=1}^{n} \int_{0}^{T} \left(u_t^i - P_T^i \right) \dd Q_t^{u,i},
	\label{eq:Y_def}
\end{align}
where $Q_t^{u,i}$ is given by either \eqref{eq:det_arrival} or \eqref{eq:Poi_arrival}.

We will denote by $J$ the bookmaker's \textit{objective function}.  We shall assume $J$ is of the form
\begin{align}
J(t,x,p,q; u)
	&=	\Eb_{t,x,p,q}  U(Y_T^u)  ,
\end{align}
where $U : \Rb \to \Rb$ is either the identity function or a utility function. 
Here, we have introduced the notation $\Eb_{t,x,p,q} \, \cdot \, = \Eb( \, \cdot \, | X_t^u = x, P_t = p, Q^u_t = q)$.

The bookmaker seeks an optimal \textit{control} or \textit{pricing policy} $u^*$ to the problem:
\begin{align}
\label{prime_prob}
V(t,x,p,q)
&:=	\sup_{u \in \Acc(t,T) } \, J(t,x,p,q; u),
\end{align}
where the admissible set $\Acc(t,T)$ is defined by \eqref{eq:Au}.
We shall refer to the function $V$ as the \textit{bookmaker's value function}. 
When we take $U$ to be the identity function, the bookmaker is risk neutral and his objective is to maximize the expected terminal wealth; see, e.g., \cite{gallego1994optimal} for the same criterion.
When we take $U$ to be an increasing and concave utility function, the bookmaker is risk-averse and his objective is to maximize the expected utility of terminal wealth.

\begin{remark}
As pointed out by an anonymous referee, the bookmakers and bettors often have different predictions about the probability of an event. 
In our framework, the (conditional) probabilities $P$ are the bookmaker's best predictions on betting outcomes. 
Note that we indirectly take into account the fact that bettors may have different views on the probabilities of outcomes and may have different objective functions.  These are all captured in a reduced form through  $\lambda_i$.
\end{remark}

\section{PDE Characterization of the Value Function}
\label{sec:PDE}

Let us denote by $\Mc$ the infinitesimal generator of the process $P$, and  
by $\d_t$, $\d_x$ and $\d_{q_i}$ the partial derivative operators with respect to the corresponding arguments.
For any $u \in \Ac$, define operator $\Lc^u$ by either of the following
\begin{align}
\Lc^u 
&:= 
\sum_{i=1}^n \lam_i(p, u) ( u_i \d_x +  \d_{q_i}  ) + \Mc 
,  \label{eq:Lc_det} \\
\Lc^u 
&:= 
\sum_{i=1}^n \lam_i(p, u) ( \theta_{u_i}^{x} \theta_1^{q_i} - 1 ) + \Mc 
,  \label{eq:Lc_poi}
\end{align} 
where $\theta_{z}^{q_i}$ is a \textit{shift operator} of size $z$ in the variable $q_i$, that is,
$\theta_z^{q_i} f(q) := f(q_1, \ldots, q_i + z , \ldots , q_n)$.
Suppose the bookmaker were to fix the prices of bets at a constant $u_t = u$.
Then $\Lc^u$ as defined in \eqref{eq:Lc_det} is the generator of $(X^u, P, Q^u)$ assuming the dynamics of $Q^u$ are described by the continuous arrivals model \eqref{eq:det_arrival}, and $\Lc^u$ as defined in \eqref{eq:Lc_poi} is the generator of $(X^u, P, Q^u)$ assuming the dynamics of $Q^u$ are described by the Poisson arrivals model \eqref{eq:Poi_arrival}.

As is standard in stochastic control theory, we provide a verification theorem 
to  Problem \eqref{prime_prob} in a general setting. 
We refer the reader to \cite{fleming2006controlled}[Section III.8] and \cite{yong1999stochastic}[Section 4.3] for proofs.

\begin{theorem}
	\label{thm:wealth_ver}
	Let $v: [0,T] \times \Rb_+ \times \Ac \times \Rb_+^n \to \Rb$ be a real-valued function which is at least once differentiable with respect to all arguments and satisfies 
	\begin{align}
	\d_x v >0  \qquad \text{ and } \qquad \d_{q_i} v < 0, \quad \forall \, i \in \Nb.
	\end{align} 
	Suppose the function $v$ satisfies the Hamilton-Jacobi-Bellman (HJB) equation 
	\begin{align}
	\label{eq:HJB}
	\d_t v + \sup_{\hat u \in \Ac} \, \Lc^{\hat u} v 
	&= 0, &
	v(T,x,p,q) 
	&=  \Eb \Big[ U \Big(x - \sum_{i = 1}^n q_i \Ib_{A_i}\Big) \Big| X_{T-}^u = x, P_{T-} = p, Q_{T-}^u = q \Big] , \quad
	\end{align}
	where $\Lc^{\hat u}$ is given by either \eqref{eq:Lc_det} or \eqref{eq:Lc_poi}.
	Then $v(t,x,p,q) = V(t,x,p,q)$ is the value function to Problem \eqref{prime_prob} 
	and the optimal price process $u^*=(u_s^*)_{s \in [t,T]}$ is given by 
	\begin{align}
	\label{eq:wealth_u1}
	u_s^* = \argmax_{\hat u \in \Ac} \; \Lc^{\hat u} v(s, X^*_s, P_s, Q_s^*). 
	\end{align}
\end{theorem}

\begin{remark}
The PDE characterization in \eqref{eq:HJB} to the value function and the optimal price process given by \eqref{eq:wealth_u1} are obtained without any assumptions on function $U$, the arrival rate/intensity function $\lam_i$, or the conditional probabilities $P$.
With additional assumptions, we may simplify \eqref{eq:HJB} and \eqref{eq:wealth_u1} to more tractable forms.
For instance, when $Q^u$ is defined by the continuous arrivals model \eqref{eq:det_arrival}, we can simplify \eqref{eq:wealth_u1} and obtain for all $i \in \Nb$ that 
\begin{align}
\label{eq:wealth_u2}
\left[\lam_i(P_s, u_s^*) + \sum_{j=1}^n u_s^{j, *} \cdot \d_{u_i} \lam_j(P_s, u_s^*) \right] \d_x V + \sum_{j=1}^n  \d_{u_i} \lam_j(P_s, u_s^*) \cdot \d_{q_j}V=0 .
\end{align}
Equation \eqref{eq:wealth_u2} reduces the problem of finding the optimal price process $u^*$ in feedback form to solving a system of $n$ equations. 
We shall apply the results of Theorem \ref{thm:wealth_ver} to obtain the value function and/or the optimal price process in closed forms in the subsequent sections.
\end{remark}

\section{Analysis of the Semi-static Setting}
\label{sec:cont}

In this section, we solve the main problem \eqref{prime_prob} in a \textit{semi-static} setting.
The standing assumptions of this section are as follows.

\begin{assumption}
	\label{assumption:con_de}
	The arrivals process $Q^{u,i}$ is given by the continuous arrivals model \eqref{eq:det_arrival}.
	The vector of conditional probabilities is a vector of constants, namely, $P_t \equiv p \in (0,1)^n$ for all $t \in [0,T)$. 	
	The utility function $U$ is continuous and strictly increasing. 
	The rate function $\lam_i=\lam_i(u^i)$ is a continuous and decreasing function of $u^i$ for all $i \in \Nb$. 
\end{assumption}

For notational simplicity, we write the rate function as $\lambda_i(u_t^i)$ 
in the rest of this section because the conditional probabilities $P_t \equiv p$ are fixed constants. 
Under Assumption \ref{assumption:con_de}, we have from \eqref{eq:Y_def} that the bookmaker's terminal wealth $Y_T^u$ is given by 
\begin{align}
\label{eq:Y_def1}
Y_T^u =X_t - \sum_{i=1}^n Q_t^i \cdot \Ib_{A_i} + \sum_{i=1}^n\int_t^T \lambda_i(u_s^i)\,u_s^i \, \dd s-\sum_{i=1}^n\int_t^T \lambda_i(u_s^i) \, \dd s\cdot \Ib_{A_i}.
\end{align}
As $\lam_i$ is decreasing by Assumption \ref{assumption:con_de}, its inverse $\lam_i^{-1}$ exists.  Defining the function $f_i$ by 
\begin{align}
\label{eq:f}
f_i(x) &:= x \cdot \lam_i^{-1}(x), & x&>0, & i &\in \Nb ,
\end{align}
we are able to rewrite $Y_T^u$ in \eqref{eq:Y_def1} as 
\begin{align}
\label{eq:Y_def2}
Y_T^u =X_t - \sum_{i=1}^n Q_t^i \cdot \Ib_{A_i} + \sum_{i=1}^n\int_t^T f_i(\lam_i(u_s^i)) \, \dd s-\sum_{i=1}^n\int_t^T \lambda_i(u_s^i) \, \dd s\cdot \Ib_{A_i}.
\end{align}
Denote by $\hat{f}_i$ the concave envelope of $f_i$ in \eqref{eq:f} and define $\hat{Y}_T^u$ by 
\begin{align}
\label{eq:Y_hat}
\hat{Y}_T^u := X_t - \sum_{i=1}^n Q_t^i \cdot \Ib_{A_i} + \sum_{i=1}^n\int_t^T \hat{f}_i(\lam_i(u_s^i)) \, \dd s - \sum_{i=1}^n\int_t^T \lambda_i(u_s^i) \, \dd s\cdot \Ib_{A_i}.
\end{align}
It is obvious that $\hat{Y}_T^u \ge Y_T^u$ for any pricing policy $u$.
We are now ready to present the main results of this section.

\begin{theorem}
	\label{thm:utility_cont_prob}
	Let Assumption \ref{assumption:con_de} hold,
	we have
	\begin{align}
	\label{eq:V1_hat}
	V(t,x,p,q) 
	&=\sup_{\hat u  \in \Ac}\Et  U\left(x - \sum_{i=1}^n q_i \Ib_{A_i} +  (T-t)\sum_{i=1}^n\hat f_i(\lambda_i(\hat{u}_i)) - (T-t)\sum_{i=1}^n\lambda_i(\hat{u}_i) \Ib_{A_i}\right)  := \hat{V}, \quad \\
	&= \sup_{ \Lam \in \Dc_\lam} \Et U\left(x - \sum_{i=1}^n q_i \Ib_{A_i} +  (T-t)\sum_{i=1}^n\hat f_i(\Lambda_i) - (T-t)\sum_{i=1}^n \Lambda_i\Ib_{A_i}\right)  ,  	
	\end{align}
	where $\Ac = [0,1]^n$ and $\Dc_\lam =(\Dc_{\lam_1}, \Dc_{\lam_2}, \cdots, \Dc_{\lam_n})$ with $\Dc_{\lam_i}$ being the range of $\lam_i$ for all $i \in \Nb$.
\end{theorem}

\begin{remark}
	\label{rem:reduction}
	The results in Theorem \ref{thm:utility_cont_prob} help us reduce the original problem, which is a dynamic optimization problem over an $n$-dimensional stochastic process, into a static optimization problem over an $n$-dimensional constant vector. 
\end{remark}

The proof of Theorem \ref{thm:utility_cont_prob} relies on the following two lemmas.

\begin{lemma}
	\label{lemma:1}
	Let Assumption  \ref{assumption:con_de} hold, we have
	\begin{align}
	\sup_{\hat u \in \Ac} \, \Et  U \left(\hat Y_T^{\hat u} \right)  =\sup_{u \in \Acc(t,T)} \, \Et  U \left(\hat Y_T^{u} \right) ,
	\end{align}
	where $\hat{Y}_T^u$ is defined by \eqref{eq:Y_hat} and set $\Acc(t,T)$ is defined by \eqref{eq:Au}.
\end{lemma}

\begin{proof}
	As the ``$\leq$'' is obvious, we proceed to show the converse inequality is also true.
	Let $u$ be any pricing policy adopted by the bookmaker. 
	As $\lam_i$ is continuous, there exists a constant $u_i^*$ such that 
	$\lambda_i(u_i^*)=\frac{1}{T-t}\int_t^T\lambda_i(u_s^i)\, \dd s$.
	By the concavity of $\hat f_i$, we obtain 
	\begin{align}
	\frac{1}{T-t}\int_t^T\hat{f}_i(\lambda_i(u_s^i))\, \dd s
	\leq\hat f_i\left(\frac{1}{T-t}\int_t^T\lambda_i(u_s^i)\, \dd s\right)=\hat f_i(\lambda_i(u_i^*)).
	\end{align}
	As a result, we have
	\begin{align}
	\Et  U(\hat Y_T^{u})   &\leq \Et U\left(X_t - \sum_{i=1}^n Q_t^i \cdot \Ib_{A_i} + (T-t)\sum_{i=1}^n\hat f_i(\lambda_i(u_i^*)) - (T-t)\sum_{i=1}^n\lambda_i(u_i^*) \Ib_{A_i}\right)  \\
	&\leq \sup_{\hat u \in \Ac}\Et U(Y_T^{\hat u}) .
	\end{align}
	Then the result follows from the arbitrariness of $u$.
\end{proof}

\begin{lemma}	
	\label{lemma:2}
	Let Assumption \ref{assumption:con_de} hold, we have
	\begin{align}
	V(t,x,p,q) \geq \sup_{\hat u \in \Ac} \, \Et  U \left(\hat Y_T^{\hat u} \right) .
	\end{align}
\end{lemma}

\begin{proof}
	For any $\eps>0$, let $u^\eps=(u_1^\eps,\dotso,u_n^\eps)$ be an $\eps/2$-optimizer for $\hat V$ in \eqref{eq:V1_hat}. Let $\delta>0$ be small enough so that 
	$\Et  \, U(\hat Y_T^{u^\eps}-\delta ) 
	> \Et  \, U (\hat Y_T^{u^\eps} )  -\eps/2
	> \hat V-\eps$.
	Next choose $v_i$, $w_i$, $\rho_i\in[0,1]$ for each $i \in \Nb$ so that the following conditions are met: 
	$\rho_i \cdot \lam_i(v_i) + (1 - \rho_i) \cdot \lam_i(w_i) = \lam_i(u_i^\eps)$
	and $\rho_i \cdot f_i( \lambda_i (v_i))+(1-\rho_i) \cdot f_i(\lambda_i(w_i)) > \hat f_i(\lambda_i(u_i^\eps))-\frac{\delta}{n (T-t)} .$
	We construct a pricing strategy $\tilde u$ by
	\begin{equation}
	\label{eq:eps_policy}
	\tilde u_s^i=v_i \cdot \Ib_{\{t\leq s<t+\rho_i(T-t)\}} + w_i \cdot \Ib_{\{t+\rho_i(T-t)\leq s\leq T\}}.
	\end{equation}
	It is easy to verify that $Y_T^{\tilde u}>\hat Y_T^{u^\eps}-\delta$,
	which in turn implies  
	$V(t,x,p,q) 
	\geq \Et  U \left(Y_T^{\tilde u} \right)
	\geq \Et \ U \left(\hat Y_T^{u^\eps}-\delta \right)  
	>\hat V-\eps$.
	Hence, the desired result follows.
\end{proof}

\begin{proof}[\textbf{Proof of Theorem \ref{thm:utility_cont_prob}}]
	The result follows from Lemmas \ref{lemma:1} and \ref{lemma:2} and the fact that $f_i\leq\hat f_i$.
\end{proof}

\begin{corollary}
	\label{cor:static}
	If $ \hat{u}^* = (\hat{u}^*_1,\hat{u}^*_2,\cdots, \hat{u}^*_n)$ is an optimizer for $\hat V$ and
	\begin{equation}
	\label{eq:f=f_hat}
	f_i(\lambda_i(\hat{u}_i^*))=\hat f_i(\lambda_i(\hat{u}_i^*)),\qquad i \in \Nb,
	\end{equation}
	then $ \hat{u}^*$ is an optimizer for $V$. 
	That is, a static strategy $u = (u_s)_{s \in [t,T]}$, where $u_s \equiv \hat{u}^*$, is an optimal strategy to Problem {\eqref{prime_prob}}.
	In general, the pricing policy $\tilde u$ defined in \eqref{eq:eps_policy} is an $\eps$-optimizer for $V$. 
\end{corollary}

\begin{remark}
	\label{rem:coin_con}
	If $f_i$ is concave, \eqref{eq:f=f_hat} is satisfied. For example, this is the case if $\lam_i$ is given by \eqref{eq:lambda-1}. 
	However, if $\lam_i$ is given by \eqref{eq:lambda-2}, $f_i$ in general is not a concave function. In fact, we have 
	\begin{align}
	f''_i(x) < 0 \quad \text{ if } 0< x < -\frac{2}{\log p_i}, \qquad \text{and} \qquad 
	f''_i(x) > 0 \quad \text{ if } x > -\frac{2}{\log p_i}.
	\end{align}
\end{remark}

As mentioned in Remark \ref{rem:reduction}, Theorem \ref{thm:utility_cont_prob} allows us to reduce the target optimization problem into a much simpler static optimization one, but the general characterization (either explicit or numerical solutions) of the optimal constant vector $\hat{u}^*$ to $\hat{V}$ is still not available. 
Below we provide an example where the optimizer $\hat u^*$ is given by a system of equations. 
In the example, we choose the rate function to be given by \eqref{eq:lambda-1}, and, by Corollary \ref{cor:static} and Remark \ref{rem:coin_con}, the static strategy $u \equiv \hat u^*$ is an optimal strategy to the bookmaker's problem.

\begin{corollary}
\label{cor:exp}
	Let Assumption \ref{assumption:con_de} hold. Assume further that (i) the sets $(A_i)_{i \in \Nb}$ 
	form a partition of $\Omega$,
	(ii) the rate function $\lam_i$ is given by \eqref{eq:lambda-1}, and (iii)  $U$ is given by  
	$U(y) = - \ee^{-\gam y}$,  where $\gam > 0$
	Then the optimizer $\hat{u}^*$ of $\hat{V}$, defined in \eqref{eq:V1_hat}, solves the following equation  
	\begin{align}
	\label{eq:new}
	p_i \cdot \left( \frac{1}{ (\hat{u}^*_i)^2} - 1\right) \cdot g(t, p_i, q_i; \hat{u}^*_i)
	= \sum_{j \neq i} p_j \cdot g(t, p_j, q_j; \hat{u}^*_j),   && \forall \, i \in \Nb,
	\end{align}
	where function $g$ is defined by 
	\begin{align}
	g(t, p_i, q_i; u_i) := \exp \left( \gam q_i + \gam (T-t) \frac{p_i}{1 - p_i} \left( \frac{1}{u_i} - 1\right) \right), 
	\end{align}
	for all $t \in [0,T)$, $p_i \in (0,1)$, $q_i \in \Rb_+$, and $ u_i \in (0,1)$.
	Moreover, if we replace assumption (i) with (i') the sets $(A_i)_{i \in \Nb}$  are independent, then $\hat{u}^*$ solves the following equation
	\begin{align}
	\label{eq:1st_exm1}
	p_i \cdot \left( \frac{1}{ (\hat{u}^*_i)^2} - 1\right) \cdot \exp \left(\gam (T-t) \frac{p_i}{1 - p_i}  \left( \frac{1}{\hat{u}_i^*} - 1\right) \right)  = 1 - p_i, && \forall \, i \in \Nb .
	\end{align}
\end{corollary}

\begin{proof}
	Equations \eqref{eq:new} and	\eqref{eq:1st_exm1} are the first-order conditions of the corresponding optimization problems under the given assumptions. 
\end{proof}

\subsection{Case Studies}

In this subsection, we  apply the results from Corollary \ref{cor:exp} to study two cases, which correspond to the two different conditions (i) and (i').

\begin{figure}[ht]
	\centering
	\includegraphics[width=0.4 \textwidth]{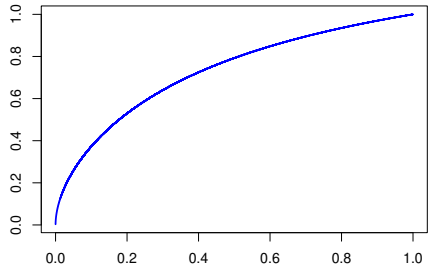}
	\\[-4ex]
	\caption{\small A plot of $\hat{u}^*$ given by \eqref{eq:1st_exm1} as a function of $p$ with $\gamma=2$ and $T-t=1$}
	\label{fig:opti}
\end{figure}

First, we study the case of independent sets $(A_i)_{i \in \Nb}$ in Corollary \ref{cor:exp}.  
According to \eqref{eq:1st_exm1}, for fixed $\gam$ and $T-t$, the optimal price $\hat u^*$ of a bet on  set $A_i$ is fully determined by the set probability $p_i$.	
In Figure \ref{fig:opti}, we plot $\hat{u}^*$ as a function of $p$, given  $\gamma=2$ and $T-t=1$.
As expected, the optimal price $\hat u^*$ is an increasing function of the outcome probability. 
In addition, $\hat{u}^*$ is a concave function of $p$, indicating that the bookmaker's optimal price is more sensitive to the changes of small probabilities. 
The economic explanation is that, the increment of bets on an event is more significant when its probability increases from 0.1 to 0.2 than that from 0.8 to 0.9. 
In other words, the changes at small probabilities cause more imbalance in the bookmaker's book than the same amount of changes but at large probabilities. 
Therefore, to counter such an effect, the bookmaker increases the price at a faster scale at small probabilities.

\begin{figure}[htb]
	\centering
	\includegraphics[width=0.85\textwidth]{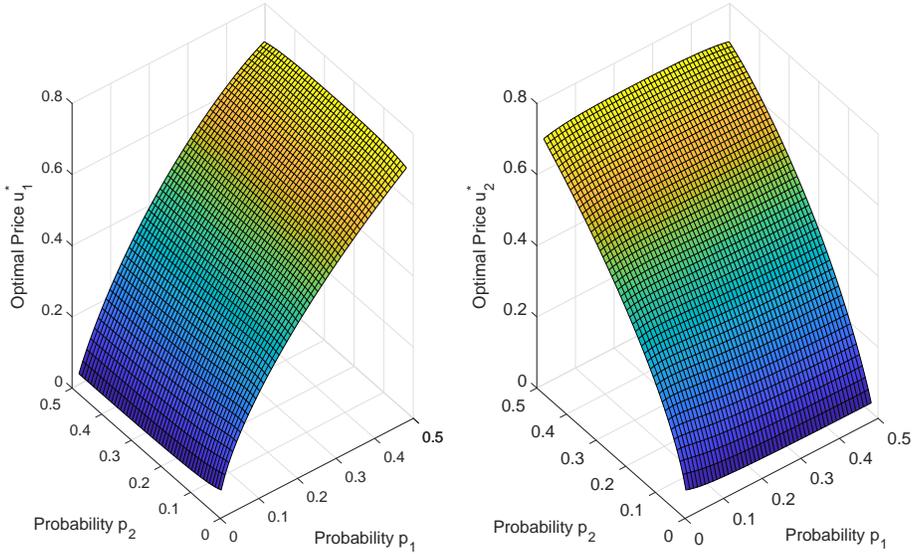}
	\\[-4ex]
	\caption{\small Surface plot of $\hat u_1^*$ and $\hat u_2^*$ with $\gamma=2$, $T - t = 5$ and $q=(0, 0, 0)$}
	\label{fig:cor_prob}
\end{figure}

\begin{figure}[htb]
	\centering  
	\includegraphics[height = 2.5in, width= 0.7 \textwidth]{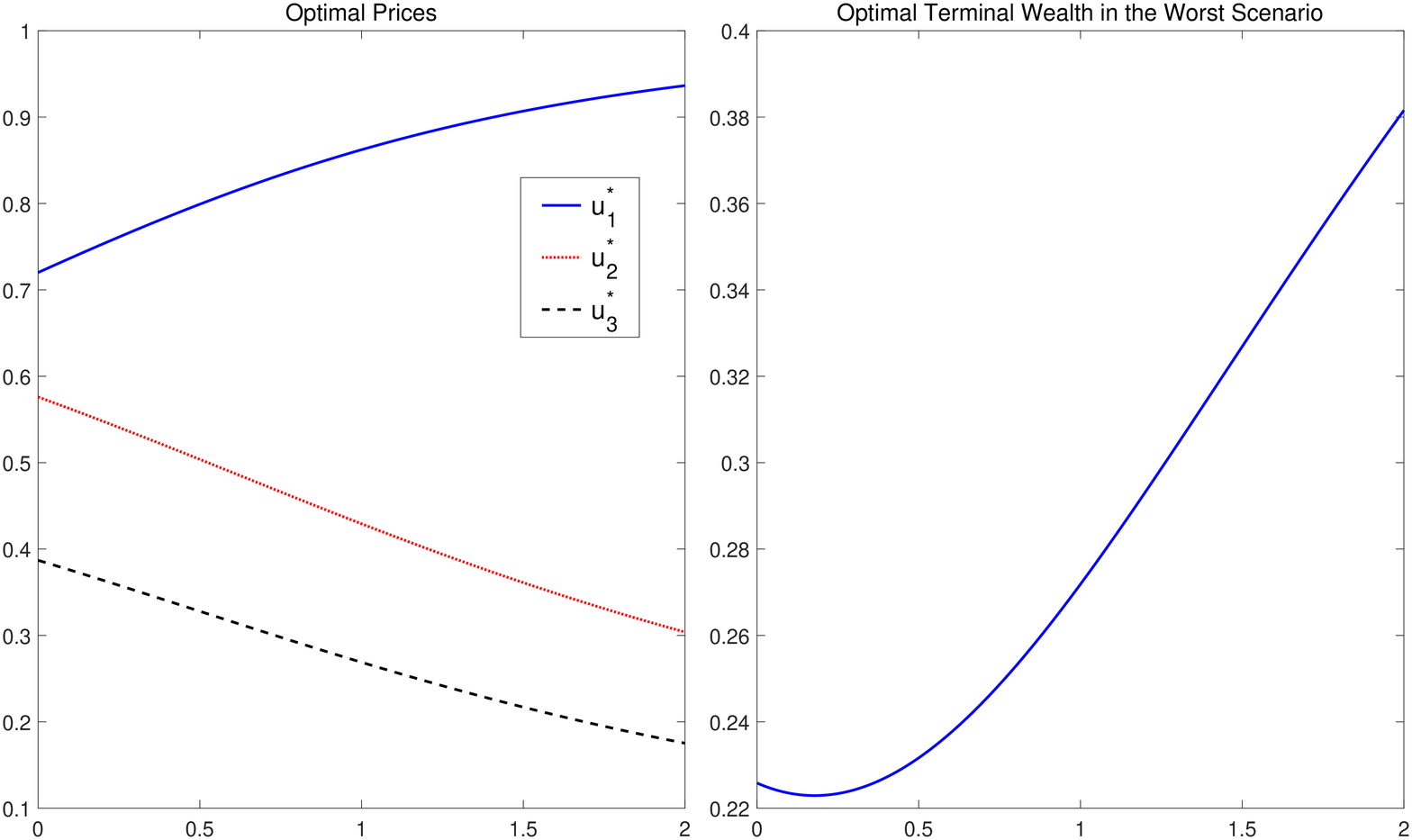}
	\caption{\small Left panel plots the optimal prices $u_i^* (= \hat u_i^*)$, $i=1,2,3$. 
			Right panel plots the bookmaker's terminal wealth after paying out the winning bets under the optimal price strategy $u^*$ in the worst scenario. In both panels, the x-axis is the current number of bets $q_1$ accumulated on outcome $A_1$, and the remaining model parameters are set as in \eqref{eq:para} with the initial wealth $x=0$.}
	\label{fig:cor_quan}
\end{figure}

Next we analyze the results of Corollary \ref{cor:exp} when the sets $(A_i)_{i \in \Nb}$ form a partition of $\Omega$. By Theorem \ref{thm:utility_cont_prob}, the optimizer $\hat u^*$ solves a static problem given in \eqref{eq:V1_hat} and is then an $n$-dimensional scalar vector independent of time $t$. It seems from \eqref{eq:new} that $\hat u^*$ may depend on $t$ through function $g$. 
However, noting the definition of $\lam_i$ in  \eqref{eq:lambda-1}, we have
\begin{align}
Q_s^{\hat u^*, i} - Q_t^{\hat u^*, i} &= \frac{p_i}{1 - p_i} \left(\frac{1}{\hat u_i^*} - 1 \right)(s-t) 
&\text{ and } & & g(t, p_i, Q_t^{\hat u^*, i}; \hat u_i^*) \equiv g(s, p_i, Q_s^{\hat u^*, i}; \hat u_i^*),
\end{align}	
which in turn shows $\hat u^*$ is indeed independent of time $t$. 
To solve \eqref{eq:new}, we consider an example with three mutually exclusive outcomes $A_1$, $A_2$ and $A_3$. We consider $p_1, p_2 \in (0, 0.5)$ and solve \eqref{eq:new} numerically given $\gamma = 2$, $T - t = 5$ and $q=(0, 0, 0)$. The surface plots of the optimal prices $\hat u_1^*$ and $\hat u_2^*$ are given in Figure \ref{fig:cor_prob}. For fixed $p_2$ (resp. $p_1$), the optimal price $\hat  u_1^*$ (resp. $\hat u_2^*$) is a strictly  increasing and concave function of the outcome probability $p_1$ (resp. $p_2$). 
From \eqref{eq:new}, with all other parameters fixed, we obtain 
$\frac{\partial \hat{u}_i^*}{\partial q_i} > 0.$
To visualize this result, we fix the model parameters as follows 
\begin{align}
\label{eq:para}
p_1 &= \frac{1}{2},  & p_2&= \frac{1}{3}, & p_3 &= \frac{1}{6}, & q_2&=q_3 =0, & \gam &= 2, & T-t &= 1,
\end{align}
but allow the current accumulated bets on outcome $A_1$ to vary from 0 to 2. Recall that $q_i$ is the current accumulated bets on outcome $A_i$ by the initial time.
We then plot the optimal prices $\hat  u_i^*$, where $i=1,2,3$, in the left panel of Figure \ref{fig:cor_quan}. 
It is readily seen that the optimal price $\hat u_1^*$ on outcome $A_1$ is an increasing function of its accumulated number of bets $q_1$. 
Such an increasing relation is consistent with our intuition that, if the bookmaker has received too many (few) bets on event $A_1$, the bookmaker should increase (decrease) the price on $A_1$ to balance the books. In fact, Figure \ref{fig:cor_quan} also reveals that the bookmaker should simultaneously decrease the prices on mutually exclusive outcomes $A_2$ and $A_3$, which further helps maintain a balanced book. The monotonic relation of $\hat u_i^*$ w.r.t.  probability $p_i$ in Figure \ref{fig:cor_quan} also confirms our findings in Figures \ref{fig:opti} and \ref{fig:cor_prob}. 
In the right panel of Figure \ref{fig:cor_quan}, we plot the bookmaker's terminal wealth (profit) after paying out the winning bets in the worst scenario (i.e., $\min \, Y_T^{\hat u^*} = X_T^{\hat u^*} - \max_{i=1,2,3} \, \{ Q_T^{\hat u^*,i}\}$) against the inventory $q_1$ on outcome $A_1$, with the initial wealth setting to $0$. Here, since outcome $A_1$ has the largest probability among the three outcomes, it also has the largest total number of bets by the terminal time $T$. In consequence, the worst scenario to the bookmaker is when outcome $A_1$ happens at time $T$, i.e., $\min Y_T^{\hat u^*} = X_T^{\hat u^*} - Q_T^{{\hat u^*},1}$. 
As shown in the right panel of Figure \ref{fig:cor_quan}, the bookmaker will pocket positive profits from the bettors in regardless of the final result of the gamble.

\section{Wealth Maximization}
\label{sec:wealth_max}

In this section, we consider Problem \eqref{prime_prob} when $U(y) = y$ (i.e., the objective of the bookmaker is to maximize his expected terminal wealth).
Henceforth, we shall refer to this problem as the \textit{wealth maximization problem}. 
We solve the wealth maximization problem using three different methods, and obtain the solutions to Problem \eqref{prime_prob} in  Theorem  \ref{thm:wealth_max} and Corollaries \ref{cor:method1}, \ref{wealth_dyna} and \ref{method3}.

\subsection{Method I}
The following theorem transforms the bookmaker's dynamic optimization problem into a static optimization problem.

\begin{theorem}
	\label{thm:wealth_max}
	Assume $U(y) = y$, and the bet arrival process $Q^{u,i}$ is given by either \eqref{eq:det_arrival} or \eqref{eq:Poi_arrival} for all $i \in \Nb$. Then we have
	\begin{align}
	V(t,x,p,q) 
	&= x - p \cdot q + \Eb_{t,x,p,q} \int_t^T \sup_{\hat u \in \Ac} \; \sum_{i=1}^n \; \lambda_i(P_s,\hat u)\, (\hat u_i-P_s^i) \, \dd s ,
	\label{eq:wealth_op}
	\end{align}
	where $p \cdot q = \sum\limits_{i=1}^n p_i q_i$.
\end{theorem}

\begin{proof}
	As $Q^{u,i}$ is given by either \eqref{eq:det_arrival} or \eqref{eq:Poi_arrival}, we have
	\begin{align}
	\Eb_t  \dd Q_t^{u,i} 
	&= \lam_i(P_t, u_t) \, \dd t,  & 
	\Eb_t  \Ib_{A_i}
	&= P_t^i,
	\end{align}
	and, as a result,
	\begin{align}
	V(t,x,p,q) 
	&= x - p \cdot q+ \sup_{u \in \Acc(t,T)} \Eb_{t,x,p,q}  \sum_{i=1}^n \left(\int_t^T u_s^i \, \dd Q_s^{u,i} -  Q_T^{u,i}\Ib_{A_i}\right)   \\
	&=  x - p \cdot q+ \sup_{u \in \Acc(t,T)} \Eb_{t,x,p,q}  \sum_{i=1}^n \left(\int_t^T \lam_i(P_s, u_s) \, (u_s^i  - P_s^i) \, \dd s \right)   \\
	&= x - p \cdot q+ \Eb_{t,x,p,q}  \int_t^T\sup_{\hat u \in \Ac} \; \sum_{i=1}^n \lambda_i(P_s,\hat u)\, (\hat u_i-P_s^i)\, \dd s , 
	\end{align}
	where the last equality can be shown by a measurable selection argument (see \cite{wagner1977survey}).
\end{proof}

\begin{remark}
	\label{rem:value}
	Due to the assumption of $U(y)=y$, the value function obtained in \eqref{eq:wealth_op} depends linearly on both the current wealth $x$ and the current books $q$. To be precise, we have
	\begin{align}
	\frac{\partial V(t,x,p,q)}{\partial x} &= 1 > 0  &\text{ and } & &
	\frac{\partial V(t,x,p,q)}{\partial q_i} &= - p_i < 0, \quad \forall \, i \in \Nb,
	\end{align}
which verify the derivatives conditions in Theorem \ref{thm:wealth_ver}.

For any $\eps > 0$, $u^\eps= (u^\eps_s(p))_{t \le s < T}=(u_s^{\eps,1}(p),\dotso,u_s^{\eps,n}(p))_{t \le s < T} $ measurably selected such that
	$$\sum_{i=1}^n \lambda_i(p,u_s^\eps(p))\cdot (u_s^{\eps,i}(p)-p_i)\geq\sup_{\hat u \in \Ac}\sum_{i=1}^n \lambda_i(p,\hat u)\,(\hat u_i-p_i)-\eps, \qquad \forall \, p \in \Ac$$
	is an $\eps$-optimizer of the value function $V$. 
\end{remark}

In light of Theorem \ref{thm:wealth_max}, we are able to reduce the complexity of the wealth maximization problem significantly.
As mentioned above, Theorem \ref{thm:wealth_max} allows us to transform a dynamic optimization problem (over $u \in \Acc(t,T)$) into a static optimization problem (over $\hat u \in \Ac$), and shows that the value function are the same for these two problems.
Under the general setup, the characterization in \eqref{eq:wealth_op} is not enough for us to obtain the optimal price process $u^*$ explicitly.
However, when the rate or intensity function $\lam_i$ is given by \eqref{eq:lambda-1} or \eqref{eq:lambda-2}, we are able to find $u^*$ in closed forms; see the corollary below.

\begin{corollary}
	\label{cor:method1}
	Assume $U(y) = y$, and the arrival process $Q^{u,i}$ is given by either \eqref{eq:det_arrival} or \eqref{eq:Poi_arrival} for all $i \in \Nb$.
	\begin{itemize}
		\item[(i)] 
		If the rate or intensity function $\lam_i$ is  given by \eqref{eq:lambda-1}, then the optimal price process to the wealth maximization problem is $u^*= (u^{1,*}_t, u^{2,*}_t, \cdots,u^{n,*}_t )_{t \in [0,T)}$, where $u_t^{i, *}$ is given by
		\begin{align}
		\label{eq:wealth_op_exm1_general}
		u^{i,*}_t = \sqrt{P_t^i}, && \forall\,  i \in \Nb.
		\end{align}
		If we assume further that $P_t \equiv p \in (0,1)^n$ for all $t \in [0,T)$, then the optimal price process $u^*$ is
		\begin{align}
		\label{eq:wealth_op_exm1}
		u^{i,*}_t \equiv \sqrt{p_i}, && \forall \,  i \in \Nb, 
		\end{align}
		and the value function $V$ is given by 
		\begin{align}
		\label{eq:dy_weal_value}
		V(t,x,p,q) = x - p \cdot q + (T-t) \, \sum_{i=1}^n \frac{p_i}{1-p_i} \left(1 - \sqrt{p_i}\right)^2.
		\end{align}

		\item[(ii)] If the rate or intensity function $\lam_i$ is  given by \eqref{eq:lambda-2}, 
		then the optimal price process to the wealth maximization problem is $u^*= (u^{1,*}_t, u^{2,*}_t, \cdots,u^{n,*}_t )_{t \in [0,T)}$, where  $u^{i,*}_t$ is the unique solution   on $(\ee^{-1},1)$ to the equation
		\begin{align}
		\label{eq:wealth_op_exm2}
		r(1+\log r)= P_t^i, && \forall \, i \in \Nb.
		\end{align}
	\end{itemize}
\end{corollary}

\begin{proof}
According to Theorem \ref{thm:wealth_max}, we have  
$u_t^{i,*} = \argmax_{\hat u_i \in [0,1]} \; \frac{1 - \hat{u}_i}{\hat{u}_i} \left(\hat{u}_i - P_t^i \right),$
when $\lam_i$ is given by \eqref{eq:lambda-1}; and 
$u_t^{i,*} = \argmax_{\hat u_i \in [0,1]} \; \log(\hat{u}_i) (\hat{u}_i - P_t^i),$
when $\lam_i$ is given by \eqref{eq:lambda-2}. The rest follows naturally.
\end{proof}

If $\lam_i$ is given by \eqref{eq:lambda-1} and $P_t \equiv p$, we obtain the value function in \eqref{eq:dy_weal_value}, from which we easily see $\frac{\partial V}{\partial x} = 1$ and $\frac{\partial V}{\partial q_i} = - p_i$ as in Remark \ref{rem:value}. In addition, we have
\begin{align}
\frac{\partial V}{\partial t} &= - \sum_{i=1}^n \frac{p_i}{1-p_i} \left(1 - \sqrt{p_i}\right)^2<0, & 
\frac{\partial V}{\partial p_i} &= - q_i - \frac{p_i + \sqrt{p_i} - 1}{ \left(1 + \sqrt{p_i}\right)^2}, \quad \forall \, i \in \Nb.
\end{align}
Due to $\frac{\partial V}{\partial t}<0$, we conclude that, all things being equal, it is always a better decision to start accepting bets earlier.
Straightforward calculus shows that $\varphi(p_i) := \frac{p_i + \sqrt{p_i} - 1}{ \left(1 + \sqrt{p_i}\right)^2}$  is an increasing function of $p_i$ and takes values in $ (-1, 1/4)$. 
When $p_i > ( (\sqrt{5}-1) / 2)^2 \approx 38\%$ or  when $q_i>1$, we have $\frac{\partial V}{\partial p_i} < 0$.

Let $\mathscr{U}: (0,1) \to (\ee^{-1},1)$ be the solution function to \eqref{eq:wealth_op_exm2}, namely, $u_t^{i, *} = \mathscr{U}(P_t^i)$. 
We easily verify that 
\begin{align}
0 < \frac{\d \mathscr{U}(z)}{\d z} = \frac{1}{2 + \log \mathscr{U}(z) } \in \left( \frac{1}{2}, 1\right) \qquad \text{and} \qquad 
\frac{\d^2 \mathscr{U}(z)}{\d z^2} = \dfrac{-1}{\mathscr{U}\big(2 + \log \mathscr{U}(z)  \big)^3} < 0.
\end{align}
Thus, the function $\mathscr{U}$ is increasing and concave.
As such, the optimal price $u^{i, *}$ on set of outcomes $A_i$ increases when the conditional probability $P^i$ increases and the rate of increase is larger when $P^i$ is small.

\subsection{Method II (Dynamic Programming Method)}

Now, we use the PDE characterization of the bookmaker's value function (Theorem \ref{thm:wealth_ver}) to solve the wealth maximization problem under the continuous arrivals model \eqref{eq:det_arrival}.  We begin the analysis by establishing some analytical properties for the value function $V$.

\begin{proposition}
	\label{prop:det_value}
	Let $U(y)=y$ for all $y \in \Rb$.
	If the value function $V(t,x,p,q)$ is differentiable with respect to $t$, $x$ and $q$, we have 
	\begin{align}
	\d_t V(t,x,p,q) <0, \qquad \d_x V(t,x,p,q)=1, \qquad \text{ and } \qquad \d_{q_i} V(t,x,p,q) = -p_i, \; \forall \, i \in \Nb.
	\end{align}
\end{proposition}

\begin{proof}
	Recalling the definition of $Y_T^u$ in \eqref{eq:Y_def},
	for any price process $u \in \Acc(t,T)$,  we obtain 
	\begin{align}
	\Eb_{t,x,p,q} Y_T^u &= \Eb_{t,x,p,q} \left[x + \sum_{i=1}^n \int_t^T u_s^i \dd Q_s^{u,i} 
	- \sum_{i=1}^n P_T^i q_i - \sum_{i=1}^n  P_T^i \int_t^T \dd Q_s^{u,i} \right] \\
	&= x - \sum_{i=1}^{n} p_i q_i + \Eb_{t,x,p,q} \sum_{i=1}^n \int_t^T (u_s^i - P_s^i ) \lam(P_s^i, u_s^i) \dd s . \label{eq:wealth_obj}
	\end{align}
	It is clear that, for the optimal price process $u^*$, we have $u_s^{i,*} - P_s^i > 0$ for all $s \in [t, T]$ and $i \in \Nb$. The desired results are then obvious.
\end{proof}

We now present the explicit solutions to the wealth maximization problem.

\begin{corollary}
	\label{wealth_dyna}
	Assume $U(y)=y$ and the bets arrive according to the continuous arrivals model \eqref{eq:det_arrival}. 
	(i) If the rate function $\lam_i$ is given by \eqref{eq:lambda-1}, then the optimal price process $u^*$ is given by \eqref{eq:wealth_op_exm1_general}. 
	(ii) If the rate function $\lam_i$ is given by \eqref{eq:lambda-2}, then the optimal price process $u^*$ is the solution to \eqref{eq:wealth_op_exm2}
\end{corollary}

\begin{proof}
(i) Under the continuous arrivals model \eqref{eq:det_arrival} with rate function \eqref{eq:lambda-1}, we derive from \eqref{eq:wealth_u1} that   
$\d_x V -\left(u_s^{i,*} \right)^{-2} \cdot \d_{q_i} V = 0$, $\forall \, i \in \Nb$.
	Together with the results from Proposition  \ref{prop:det_value}, we obtain, for all $s \in [t,T)$, that 
	$u_s^{i,*} = \sqrt{ -\d_{q_i} V / \d_x V} \equiv \sqrt{P_s^i} \in (0,1)$, $\forall \, i \in \Nb$.
(ii) Under the continuous arrivals model \eqref{eq:det_arrival} with rate function \eqref{eq:lambda-2}, we derive from \eqref{eq:wealth_u1} that 
$(1 + \log u_t^{i,*} ) \, \d_x V + \d_{q_i} V / u_t^{i,*}=0$, $\forall \, i \in \Nb$.
The desired results are then obtained.
\end{proof}

\subsection{Method III}

Lastly, we use the results from Theorem \ref{thm:utility_cont_prob} to solve the wealth maximization problem.

\begin{corollary} 
\label{method3}
	Let Assumption  \ref{assumption:con_de} hold and suppose $U(y) = y$. If an interior optimizer $\Lam^*$ in Theorem \ref{thm:utility_cont_prob} exists, then we have
	\begin{align}
	\hat{f}_i'(\Lam_i^*) = p_i, && \forall \, i \in \Nb.
	\end{align} 
	In particular, if $\lam_i$ is given by \eqref{eq:lambda-1},  we obtain
	\begin{align}
	\label{eq:uti_wealth_op}
	\Lam_i^* = \frac{p_i}{1 - p_i} \left( \frac{1}{\sqrt{p_i}} - 1\right) \qquad \text{and} \qquad 
	\hat{u}^*_i = \sqrt{p_i} , && \forall \, i \in \Nb 
	\end{align}
	and if $\lam_i$ is given by \eqref{eq:lambda-2}, then
	\begin{align}
	\label{eq:uhat_lam2}
	p_i^{\Lam_i^*} \left(1 + \Lam_i^* \log p_i \right) = p_i \qquad \text{and} \qquad 
	\hat{u}^*_i (1 + \log \hat{u}^*_i) = p_i, && \forall \, i \in \Nb.
	\end{align}
	
\end{corollary}

\begin{proof}
	The first general result is immediate thanks to Theorem \ref{thm:utility_cont_prob} and the assumption of $U(y) = y$. As pointed out in Remark \ref{rem:coin_con}, 
	if $\lam_i$ is given by \eqref{eq:lambda-1}, we have $f_i = \hat{f}_i$.
	If $\lam_i$ is given by \eqref{eq:lambda-2}, we calculate $\hat{f}'_i(x) = f_i'(x) >0$ for $0 < x < - \frac{1}{\log p_i}$.
\end{proof}

\subsection{Comparison of the Three Methods}

In terms of the model generality, Theorem \ref{thm:wealth_max} obtained via Method I is the most general one, since (i) the conditional probabilities $P$ can be modeled by any stochastic process taking values in (0,1) and (ii) the arrival process is given by either the continuous arrival model \eqref{eq:det_arrival} or the Poisson arrivals model \eqref{eq:Poi_arrival}. 
Corollary \ref{wealth_dyna} of Method II holds when the arrival process is given by the continuous arrival model \eqref{eq:det_arrival}. 
Corollary \ref{method3} of Method III is further restricted to $P$ being constants.

In terms of application scope, Method I is the most restricted one, as it only applies to the wealth maximization problem. Both Methods II and III are developed to solve the main problem for a general function $U$, and hence can be applied to solve concave utility maximization problem (see for instance Corollary \ref{cor:exp}).

Both Theorem \ref{thm:wealth_max} of Method I and Theorem \ref{thm:utility_cont_prob} of Method III provide an explicit characterization to the value function, but do not provide an explicit solution to the optimal price process. Method II (dynamic programming method) provides characterizations to both the value function and the optimal price process.  However, solving the HJB equation \eqref{eq:HJB} in Method II  is often a challenging task.

From a computational point of view, the static optimization problem in Method I (see \eqref{eq:wealth_op}) or in Method III (see \eqref{eq:V1_hat}) is easy to solve, while finding the numerical solutions to the HJB equation \eqref{eq:HJB} under the feedback strategy \eqref{eq:wealth_u1} may be difficult computationally. 


\subsection{Numerical Analysis of a Coin Example}
\label{subsec:prob}

In this subsection, we consider a simple coin tossing example, which has two mutually exclusive outcomes $A_1 = \{\text{Heads}\}$ and $A_2 = \{\text{Tails}\}$. Suppose $P_t^1 \equiv \hat p \in (0,1)$ for all $t \in [0,T)$, and then $P_t^2 \equiv 1 - \hat p$. 
We carry out the analysis based on the two models of the arrival process $Q^{u^*,i}$ -- the continuous arrivals model \eqref{eq:det_arrival} and the Poisson arrivals model \eqref{eq:Poi_arrival}. 
The rate or intensity function $\lam_i$ is given by \eqref{eq:lambda-1}. 
In addition, we shall assume that the bookmaker has taken zero bets at time $t=0$ and has zero initial wealth:
\begin{align}
P_t
	&=	(\hat p, 1 - \hat p) , & 
\forall \, t
	&< T , &
X_0
	&=	0 , &
Q_0^i
	&=	0 , &
\forall \, i
	&\in \Nb . \label{eq:ICs}
\end{align}
Under these conditions, the bookmaker's optimal policy $u^*$ is given by \eqref{eq:wealth_op_exm1}.
In what follows, we analyze the probability that the bookmaker makes profits when he follows the optimal price process, denoted by $\Pb(Y_T^* >0)$.
\\[0.5em]
\textbf{Case 1:} Continuous arrivals model \eqref{eq:det_arrival} with $\lam_i$ given by \eqref{eq:lambda-1}.
In this case, the bookmaker's terminal wealth $Y^*_T$ is given by 
	\begin{align}
	Y^*_T (\text{Heads}) = \psi_1(\hat p) \cdot T \qquad \text{and} \qquad Y^*_T (\text{Tails}) = \psi_1(1-\hat p) \cdot T,
	\end{align}
	where function $\psi_1$ is defined by 
	\begin{align}
	\label{eq:psi_1}
	\psi_1(\hat{p}) 
		&:= \frac{\hat{p}}{1 - \hat{p}} \left( 2 - \sqrt{\hat{p}} - \frac{1}{\sqrt{\hat{p}}} \right) + \frac{1-\hat{p}}{\hat{p}} (1 - \sqrt{1-\hat{p}}), \qquad \hat{p} \in (0,1).
	\end{align}
	It can be shown that $\psi_1$ is a decreasing function over $(0,1)$ with  
	$\lim_{\hat{p} \to 0} \psi_1(\hat{p}) = \frac{1}{2} $ and $\lim_{\hat{p} \to 1} \psi_1(\hat{p}) = 0$.
Hence, regardless of the probability that the coin toss results in a head, the bookmaker is guaranteed to make a profit by following the optimal price policy given by \eqref{eq:wealth_op_exm1}. If $\hat p=0.5$ (the coin is fair), the bookmaker's profit is a constant given by
	\begin{align}
	Y_T^* = \psi_1(0.5) \cdot T \approx 0.171573 \cdot T.
	\end{align}	
		

\noindent
\textbf{Case 2.} Poisson arrivals model \eqref{eq:Poi_arrival} with $\lam_i$ given by \eqref{eq:lambda-1}.
In this case, the bookmaker's terminal wealth $Y^*_T$ is given by 
\begin{align}
Y^*_T (\text{Heads}) &= (\sqrt{\hat p}-1) \cdot Q_T^{u^*,1} + \sqrt{1-\hat p} \cdot Q_T^{u^*,2}, \\ 
Y^*_T (\text{Tails}) &= \sqrt{\hat p} \cdot Q_T^{u^*,1} + (\sqrt{1-\hat p} - 1) \cdot Q_T^{u^*,2},
\end{align}
where the total number of bets placed on Heads and Tails $Q_T^{u^*,i}$ $(i=1,2)$ are independent Poisson random variables with expectations given by 
$\lam_i^* \cdot T= \frac{\sqrt{p_i} (1- \sqrt{p_i})}{1 - p_i} \cdot T$, where $p_1= \hat p$ and  
$p_2 = 1-\hat p$ 
Let us compute $\Pb(Y_T^* > 0)$, the probability that the bookmaker obtains a profit.  We have 
\begin{align}
\Pb(Y_T^* > 0) &= \Pb(\text{Heads}) \Pb\left(  Q_T^{u^*,1} < \dfrac{\sqrt{1-\hat p}}{1 - \sqrt{\hat p}} \cdot Q_T^{u^*,2} \right) + \Pb(\text{Tails})  
\Pb \left( Q_T^{u^*,2} <  \dfrac{\sqrt{\hat p}}{1 - \sqrt{1-\hat p}}\cdot  Q_T^{u^*,1}   \right) \\
&= \hat{p} \sum_{j = 1}^\infty \ee^{-\lam_2^*T} \dfrac{(\lam_2^*T)^j}{j!} \sum_{i=0}^{\Mc_1(j)} \ee^{-\lam_1^*T} \dfrac{(\lam_1^*T)^i}{i!}  + 
(1-\hat p) \sum_{i = 1}^\infty \ee^{-\lam_1^*T} \dfrac{(\lam_1^*T)^i}{i!}\sum_{j=0}^{\Mc_2(i)} \ee^{-\lam_2^*T} \dfrac{(\lam_2^*T)^j}{j!} , 
\end{align}
where $\Mc_1(k)$ is the largest integer less than $\frac{\sqrt{1-\hat p}}{1 - \sqrt{\hat p}} \cdot k$ and $\Mc_2(k)$ is the largest integer less than $\frac{\sqrt{\hat p}}{1 - \sqrt{1-\hat p}} \cdot k$ for all $k=0,1,\cdots$.
The above expression of $\Pb(Y_T^* > 0) $ allows us to compute $\Pb(Y_T^* > 0)$ numerically in an efficient way. 
For instance,  given $\hat p = 0.5$, we compute 
\begin{align}
\label{eq:prob_Y_coin}
\Pb(Y_T^* > 0)  = 33.67\%  \; (T=1); \quad 54.43\% \; (T=2); \quad 76.82\% \; (T=5); \quad 86.49\%  \; (T=10).
\end{align}

\subsection{Numerical Analysis of a Diffusion Example}
\label{sub:diff_exm}

In this section, we consider a diffusion example, in which the conditional probabilities are given by diffusion processes and the betting arrivals are given by the Poisson model \eqref{eq:Poi_arrival}. 
The arrival intensity $\lam_i$ is assumed to take the form in \eqref{eq:lambda-1}.
Our main focus is to investigate the bookmaker's terminal payoff $Y_T^*$, after paying out the winning bets, when he follows the optimal price strategy $u^*$ given by \eqref{eq:wealth_op_exm1_general}.

In particular, we follow the setup in Example \ref{ex:n-points} and consider a spreads betting on an NBA game, which offers three outcomes to bet on: $A_1 = \{\text{home team wins by 3 points or more} \}$, 
$A_2 = \{\text{home team wins by less than 3 points}\}$, and $A_3 = \{\text{away team wins}\}$. 
The reason behind such a choice of outcomes is that, home advantage is well documented and empirically tested in the sports literature, and the latest home advantage is 2.33 (points) per game during the 2019-2020 NBA season.\footnote{See statistics here \url{https://www.usatoday.com/sports/nba/sagarin/}.} 
Let $\Del=(\Del_t)_{0 \le t \le T}$ denote the scores of the home team minus the scores of the away team, and recall from Example \ref{ex:n-points} that $\Del_t = \mu t + \sig W_t$.
We have 
\begin{align}
P_t^1 &= \Eb_t \Ib_{A_1} = \Pb(\Del_T \ge 3 | \Del_t), &
P_t^2 &= \Eb_t \Ib_{A_2} = \Pb(\Del_T \ge 0 | \Del_t) - P_t^1, & 
P_t^3 &= 1 - \Pb(\Del_T \ge 0 | \Del_t),
\end{align}
where $\Pb(\Del_T \ge n | \Del_t)$ is given by \eqref{eq:prob_diff}. 
We set the model parameters by 
\begin{align}
\label{eq:para_diff}
\mu &= 2.33, & \sig &= 10, & 
x &= 0, & q &=(0,0,0), &
T &= 1,  & \kappa &= 10000,
\end{align}
where $\kappa$ is the scaling factor of the intensity function in \eqref{eq:lambda-1}. 
Since we normalize $T=1$, the factor of $\kappa = 10000$ predicts that the average number of bets is about 4142 if the probability of the outcome is 0.5 and the price is set to $\sqrt{0.5}$.
We comment that our objective is to obtain \emph{qualitative}, not \emph{quantitative}, finding, since real market data are business secretes and are not available to us. The findings are robust and hold under any positive scaling factor $\kappa$.

\begin{figure}[h!]
	\centering
	\includegraphics[width=0.6\textwidth]{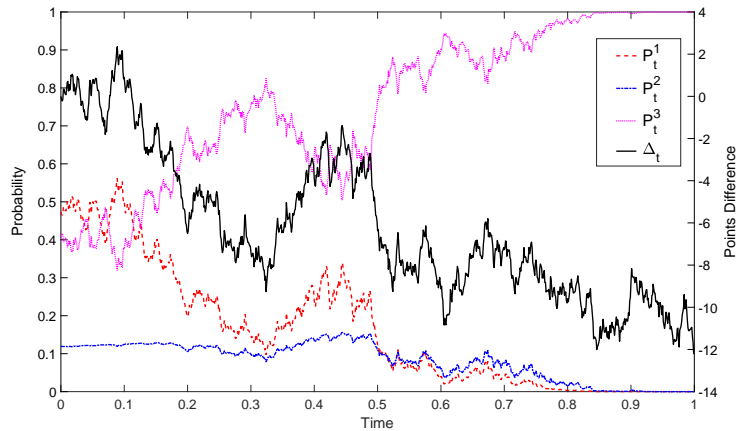}
	\\[-2ex]
	\caption{\small A simulated path of the points difference $\Del$ and the conditional probabilities $P^i$, $i=1,2,3$, over time $t \in [0,1]$. The $y$-axis on the left is probability, while the $y$-axis on the right is points difference $\Del$. The parameters are set in \eqref{eq:para_diff}.}
	\label{fig:path}
\end{figure}	
We first draw one simulated sample path of the points difference $\Del$ and compute the conditional probabilities $P_t^i$ for $i=1,2,3$ over time $t \in [0,1]$ in Figure \ref{fig:path}. 
In this particular path, the points difference (black solid line) eventually goes to negative, meaning the away team wins the game and outcome $A_3$ is the winning bet. Accordingly, we observe that the conditional probabilities $P_t^1$ (red dashed line) and $P_t^3$ (blue dash-dot line) converge to 0, and $P_t^3$ (pink dotted line) converges to 1. 
Recall $u_t^{*,i} = \sqrt{P_t^i}$, the dynamic movements of the conditional probabilities imply that the bookmaker indeed should adjust the prices dynamically to account for the changes of $\Del_t$. 
Along the particular path in Figure \ref{fig:path}, we compute 
\begin{align}
X_T^{*} &= 5471, & Q_T^{*,1} &= 2356, & Q_T^{*,2} &= 2032, & Q_T^{*,3} &= 4323, 
\end{align}
which shows the final profit of the bookmaker is $Y_T^* = X_T^* - Q_T^{*,3} = 1148 > 0$.

Next, we use the Monte Carlo method and repeat the same simulation process $N$ times (we take $N = 10,000$). We report the descriptive statistics of $Y_T^*$ in Table \ref{tab:profit}.
The key finding is that, regardless which outcome turns out to be the winning one, the bookmaker is able to make a positive profit when following the optimal dynamic price policy $u^*$. 
The median of $Y_T^*$ is 3134, greater than the mean $2433$, indicating the distribution is left-skewed. 
Comparing this example with the coin example considered in Section \ref{subsec:prob}, we see the importance of dynamically adjusting the prices for the bookmaker. 
In the coin example, the conditional probabilities remain unchanged and the bookmaker's optimal price strategy is a constant policy. In such a case, the probability of profitability is strictly less than 1; see \eqref{eq:prob_Y_coin}. 
However, in the example considered here, the dynamic movements of the points difference lead to dynamic changes of the conditional probabilities. The bookmaker then adjusts its offering on odds in a dynamic manner, and is able to achieve strict profitability by following the optimal strategy.

\begin{table}[h]
\centering
\begin{tabular}{ccccccc} \hline
mean & s.d.  & min & Q25\% & median & Q75\% & max \\ \hline 
2433 & 1048  & 163 & 1400 & 3134 & 3322 & 4035 \\ \hline
\end{tabular}
\caption{\small Descriptive statistics of the bookmaker's profit $Y_T^*$. Here, s.d. is standard deviation, and Q25\% and Q75\% refer to 25\% and 75\% quantiles, respectively. The parameters are set in \eqref{eq:para_diff}.}
\label{tab:profit}
\end{table}

\section{Exponential Utility Maximization}
\label{sec:exp}

In this section, we study Problem \eqref{prime_prob}  when $U(y)=-\ee^{-\gamma y}$, with $\gam>0$, henceforth called the \emph{exponential utility maximization problem}. 
We solve this problem assuming $Q^u$ is described by the Poisson arrivals model \eqref{eq:Poi_arrival}.\footnote{Notice Corollary \ref{cor:exp} solves the exponential utility maximization problem under the continuous arrivals model \eqref{eq:det_arrival}, among other model assumptions.}
We summarize the key results in Theorem \ref{thm:exp_dyna_v1}.

\subsection{Main Results}
Let us begin by stating some assumptions, which are assumed to hold throughout the analysis that follows.

\begin{assumption}
\label{assumption:exp}
The utility function $U$ is given by $U(y)=-\ee^{-\gamma y}$, where $\gam>0$. 
The vector of conditional probabilities is a vector of constants, namely, $P_t \equiv p \in (0,1)^n$ for all $t \in [0,T)$. 
The bet arrivals process $Q^{u,i}$ is given by the Poisson arrivals model \eqref{eq:Poi_arrival}.
The intensity function $\lam_i$ is given by \eqref{eq:lambda-3}, i.e., $\lam_i(p_i,u_i) = \kappa \ee^{-\beta(u_i - p_i)}$, where $\kappa, \beta>0$. 
\end{assumption} 

As in the previous section, because the conditional probability of set $A_i$ is assumed to be a fixed constant $p_i$, in order to simplify the notation, we write 
the intensity function as $\lam_i(u^i)$ for all $i \in \Nb$. 
Under Assumption \ref{assumption:exp}, with $X_t=x$, $P_t=p$ and $Q_t = q$, the bookmaker's terminal wealth $Y^u_T$ is given by
\begin{align}
Y_T^u = \left(x+\sum_{i=1}^n\int_t^T u_t^i \dd N_t^{u,i}-\sum_{i=1}^n\left(q_i+\int_t^T \dd N_t^{u,i}\right) \Ib_{A_i}\right).
\end{align}
The standard method of solving a stochastic control problem is to develop a verification theorem first, solve the associated HJB PDE with appropriate boundary conditions, and verify all the conditions in the verification theorem are met. We have followed exactly this standard method previously; see Theorem \ref{thm:wealth_ver} and Corollary \ref{wealth_dyna}. 
It is clear that the general verification theorem obtained in Theorem \ref{thm:wealth_ver} also applies to the problem we are considering in this section, with operator $\Lc^u$ given by \eqref{eq:Lc_poi} and $\Mc=0$.
Please refer to \cite{oksendal2005applied}[Chapter 3] for standard stochastic control theory with jumps.

The HJB associated with the exponential utility maximization problem is
\begin{equation}\label{e7}
\d_t V(t,x,p,q)+\sum_{i=1}^n\sup_{u_i}\left(\lambda_i(u_i)[V(t,x+u_i,p,q+ \e_i)-V(t,x,p,q)]\right)=0,
\end{equation}
and the boundary condition is
\begin{align}
\label{eq:boun1_V4}
V(T,x,p,q)=-\ee^{-\gamma x} \cdot a(q) ,
\end{align}
where $\e_i\in\N^n$ is the vector whose $i$-th component is $1$ and other components are $0$, and
\begin{align}
a(q):=\Eb \exp\left(\gamma\sum_{i=1}^n q_i \Ib_{A_i}\right) .
\end{align}
To better present the solutions in Theorem \ref{thm:exp_dyna_v1}, we define constants $c$, $b_i$ and $h_i$ by 
\begin{align}
\label{eq:cont_bc}
c:=\frac{\beta}{\gamma}, \qquad 
b_i := \kappa \ee^{\beta p_i}\cdot\frac{\gamma}{\beta+\gamma}\cdot\left(\frac{\beta}{\beta+\gamma}\right)^{\frac{\beta}{\gamma}},  \qquad \text{ and } \qquad 
h_i := c b_i,
\qquad i \in \Nb,
\end{align}
and functions $d(q)$ and $\alpha_k(q)$, $k=0,1,2,\cdots$, by 
\begin{align}
\label{eq:d_def}
d(q) &:=  [a(q)]^{-c},  &
\alpha_0(q) &:= d(q),  &
\alpha_k(q) &:= \frac{1}{k!}\sum_{\j\in\N^n:|\j|=k} \left(\prod_{i=1}^{n} h_i^{j_i}\right) \cdot d(q+\j), 
\end{align}
where $\j=(j_1,\dotso,j_n)\in\N^n$ and $|\j|:=j_1+\dotso+j_n$. For instance, when $k=1$, we have 
$\alpha_1(q) := \sum_{i = 1}^n \, h_i \cdot d(q+\e_i)$.

\begin{theorem}
\label{thm:exp_dyna_v1}
Let Assumption \ref{assumption:exp} hold. 
The value function $V$ of Problem \eqref{prime_prob} is given by 
\begin{align}
\label{eq:value_V4}
V(t,x,p,q) 
= - \ee^{-\gam x} \, \left[G(t,q) \right]^{- 1/c},
\end{align}
where function $G$ is defined by 
\begin{align}
\label{eq:hat_V4}
G(T,q) &=  d(q), &
G(t,q) &= \sum_{k=0}^{\infty} \alpha_k(q) \cdot (T-t)^k, \qquad t \in [0,T).
\end{align}
The optimal price process $u^*=(u_s^*)_{s \in [t,T]}$ to Problem \eqref{prime_prob} is given by 
\begin{align}
\label{eq:op_exp_dyna}
u_s^{i,*} = - \frac{1}{\gam}\log\left[\dfrac{\beta \cdot H(s,Q_s^*)}{(\beta+\gamma) \cdot 
	H(s,Q_s^*+\e_i)}\right], \qquad i \in \Nb,
\end{align}
where $H(t,q) := [G(t,q)]^{- 1/c}$.
\end{theorem}

\begin{proof}
As the utility function is of exponential form, we make the following Ansatz
\begin{align}
V(t,x,p,q) = - \ee^{-\gamma x} \cdot H(t,q).
\end{align}
Inserting the Ansatz into the HJB equation \eqref{e7} and the boundary condition \eqref{eq:boun1_V4}, we obtain 
\begin{align}\label{e4}
\d_t H(t,q)+\sum_{i=1}^n\inf_{u_i}\left(\lambda_i(u_i)[\ee^{-\gamma u_i} H(t,q+\e_i)- H(t,q)]\right)=0,\qquad \text{and} \qquad H(T,q)=a(q).
\end{align}
Solving the infimum problems in \eqref{e4} gives
\begin{align}\label{e9}
u_i^*=-\frac{1}{\gamma}\log\left[\frac{\beta \cdot H(t,q)}{(\beta+\gamma) \cdot H(t,q+\e_i)}\right]>0,
\end{align}
which proves the optimal price process in \eqref{eq:op_exp_dyna} once $H$ is found.
Next, substituting $u_i^*$ in \eqref{e9} back into \eqref{e4}, we obtain
\begin{align}\label{e8}
\d_t H(t,q)-\sum_{i=1}^n b_i\cdot\frac{[H(t,q)]^{c+1}}{[H(t,q+\e_i)]^c}=0, 
\end{align}
where constants $b_i$ and $c$ are defined by \eqref{eq:cont_bc}. Now let
$G(t,q):=[H(t,q)]^{-c}.$
We establish the equation for $G$ as follows
\begin{align}\label{e5}
\d_t G(t,q) + \sum_{i=1}^n h_i \cdot G(t,q+\e_i)=0,\qquad \text{and} \qquad
 G(T,q)=d(q).
\end{align}
where $h_i=cb_i$ as in \eqref{eq:cont_bc}.
The results from Lemma \ref{lem:sol_hat_V4} verify that $G$ given by \eqref{eq:hat_V4} solves \eqref{e5}, and therefore the value function to Problem \eqref{prime_prob} is given by \eqref{eq:value_V4}. 
\end{proof}

\begin{lemma}
\label{lem:sol_hat_V4} 
We have 
\begin{enumerate}
	\item The series $x\mapsto \sum_{k=0}^{\infty} \alpha_k(q) x^k$ is convergent with radius $+\infty$, and hence,  
	$G(t,q)$ given by \eqref{eq:hat_V4} is well defined and is in fact analytic. 
	
	\item $G(t,q)$ given by \eqref{eq:hat_V4} is a solution to \eqref{e5}.
\end{enumerate}
\end{lemma}

We place the proof to Lemma \ref{lem:sol_hat_V4} in online Appendix, along with a technical corollary that deals with the scenarios when the bookmaker sets an upper bound on the number of bets $Q$.

\subsection{Sensitivity Analysis}
\label{sub:sen}

In this subsection, we focus on the sensitivity analysis of  the value function and the optimal strategy on the model parameters. 
By Theorem \ref{thm:exp_dyna_v1}, we obtain that 
\begin{align}
\frac{\partial V}{\partial t} &= -\ee^{-\gam x} \, \left[G(t,q) \right]^{- 1/c - 1} \, \left(\sum_{k=1}^\infty \, k \, \alpha_k(q) \, (T-t)^{k-1} \right) < 0, \\
\frac{\partial V}{\partial x} &= \gamma \ee^{-\gam x} \, \left[G(t,q) \right]^{- 1/c}>0,
\end{align}
and 
\begin{align}
\frac{\partial d(q)}{\partial q_j} &= -c \big[a(q) \big]^{-c-1} \, \Eb \left[\gam \Ib_{A_j} \sum_{i=1}^n \, q_i \Ib_{A_i} \right]  & &\Rightarrow& \frac{\partial \alpha_k(q)}{\partial q_j} &<0
 & &\Rightarrow&  & \frac{\partial V}{\partial q_j} &<0,
\end{align}
which are consistent with the findings in the previous section; see Remark \ref{rem:value}, the comments on Corollary \ref{cor:method1}, and Proposition \ref{prop:det_value}. 
In deriving Theorem \ref{thm:exp_dyna_v1}, we assume the intensity function $\lam_i$ is given by \eqref{eq:lambda-3}, i.e.,  $\lam_i(p_i, u_i) = \kappa \ee^{-\beta (u_i - p_i)}$, where $\beta$ is the key parameter while $\kappa$ is merely a scaling factor.
Indeed, $\kappa$ affects the constants $b_i$ and $h_i$ proportionally, and we easily obtain $\partial V / \partial \kappa >0$ as expected. 
However, the impact of $\beta$ on the value function is rather complex, and no analytical result is available, noting $\beta$ affects both $c$ and $G$.
Our conjecture is that the bookmaker benefits more from a smaller $\beta$, as the smaller the $\beta$, the more the expected number of bets.

\begin{figure}[h!]
	\centering
	\hspace{-2ex}
	\includegraphics[width=0.48 \textwidth]{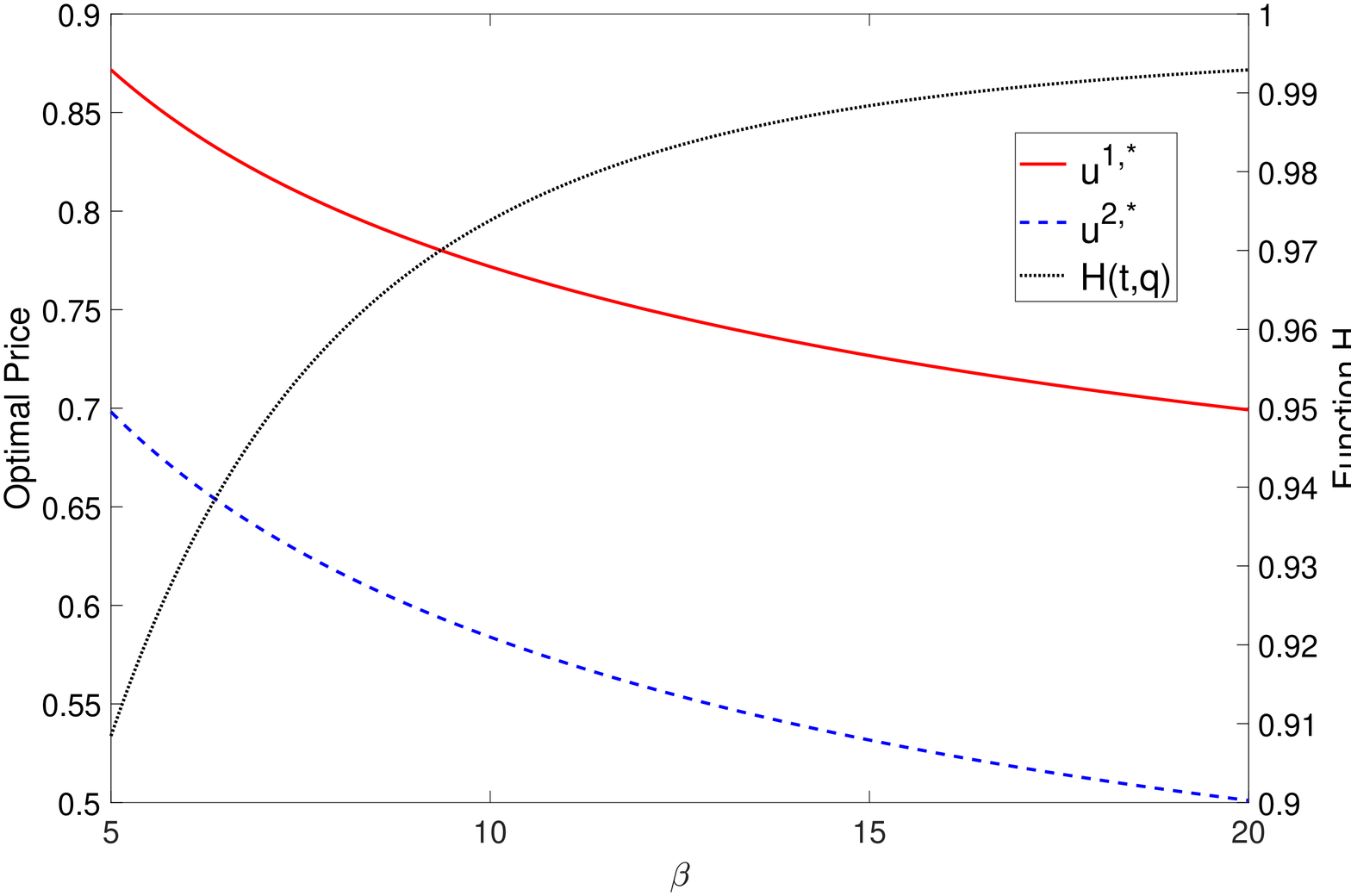}
	\includegraphics[width=0.48 \textwidth]{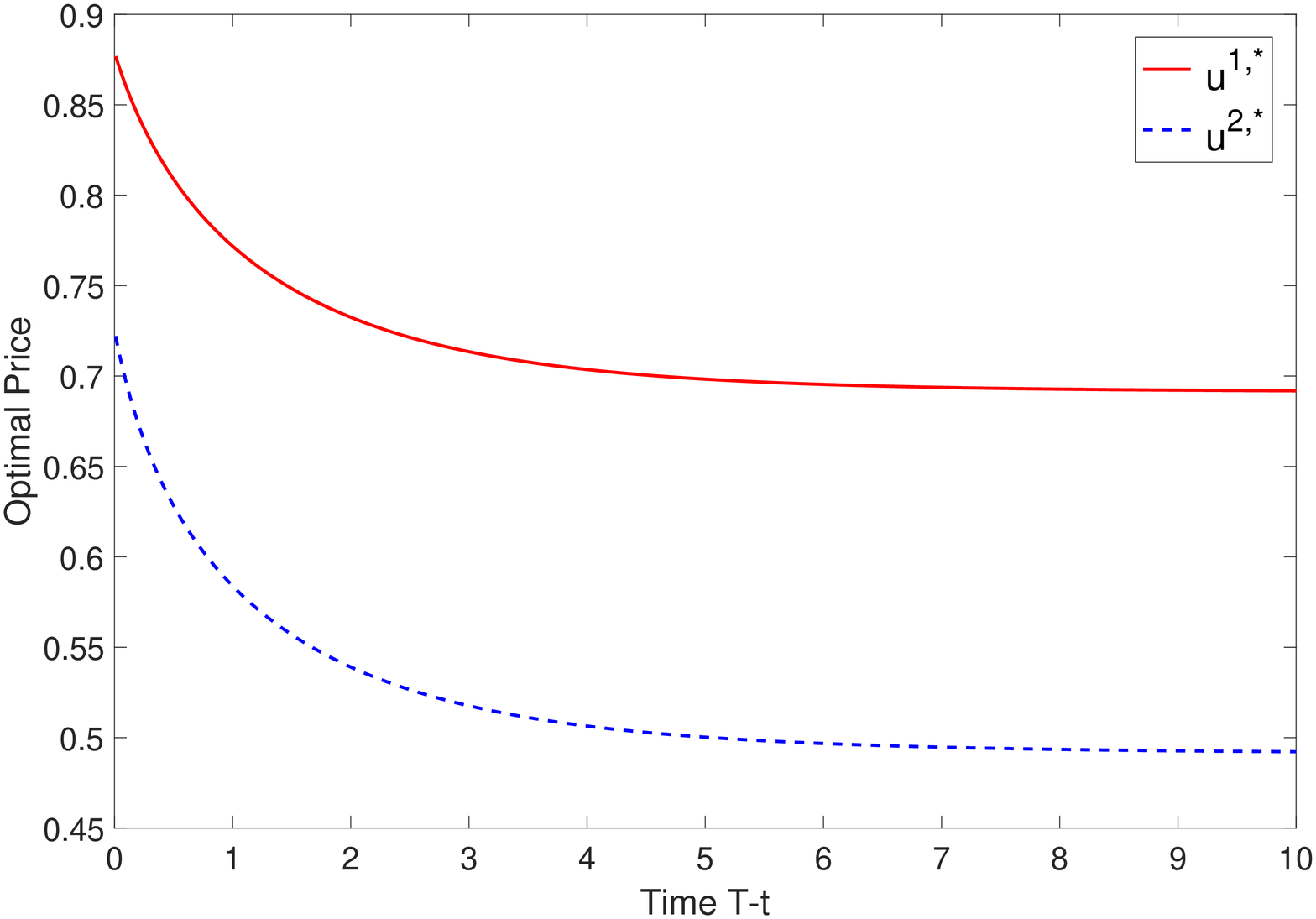}
	\\[-2ex]
	\caption{\small Left panel plots the impact of $\beta$ on the optimal price strategies $u^{1,*}$ (red solid line) and  $u^{2,*}$ (blue dashed line), and  the function $H(t,q)$ (black dotted line). The left $y$-axis is for optimal prices and the right $y$-axis is for the function $H$. 
	Right panel plots the impact of the remaining time $T-t$ on the optimal price strategies $u^{1,*}$ (red solid line) and  $u^{2,*}$ (blue dashed line). 
	All the parameters other than $\beta$ (left) or $T-t$ (right)  are set as in \eqref{eq:para_exp}.}
	\label{fig:exp_beta}
\end{figure}

To continue the investigations, we consider a coin example as in Section \ref{subsec:prob}, and denote $\Pb(A_1) = \Pb(\text{Heads}) = \hat p \in (0,1)$ and $\Pb(A_2) = \Pb(\text{Tails}) = \hat p$. 
We set the base parameters  as follows:
\begin{align}
\label{eq:para_exp}
\hat p &= 0.6,  & q_1&= q_2 = 0,  & T-t&=1, & \beta &= 10, & \kappa &= 1, & \gamma &= 2.
\end{align}
In the sensitivity analysis, we will allow the parameter of the study to vary around its base value while fix the values of all other parameters as in \eqref{eq:para_exp}. 
We first study how $\beta$ affects the bookmaker's price strategies, which in turn helps answer the impact of $\beta$ on the value function.
We comment that the bookmaker controls the prices (intensities) to attract more bets and to balance the book, which are the key to understand the following results.
We draw the plot of the optimal price strategies $(u^{1,*}, u^{2,*})$ and the function $H(t,q)$ over $\beta \in [5,20]$ in the left panel of Figure \ref{fig:exp_beta}.
As $\beta$ increases, the expected number of bets on both outcomes $A_1$ and $A_2$ decreases, and, to attract more bettors, the bookmaker lowers the prices on both $A_1$ and $A_2$, as observed in Figure \ref{fig:exp_beta}. 
Recall $V = - \ee^{-\gam x} \, H(t,q)$, the fact that $H$ is an increasing function of $\beta$ implies that $\partial V / \partial \beta < 0$, which means a decrease in the amount of bets is adverse to the bookmaker and hence verifies our previous conjecture. 
We next study the impact of the remaining time $T-t$ on the optimal prices and obtain the results in the right panel of Figure \ref{fig:exp_beta}. 
When the bookmaker has enough time (i.e., large $T-t$) to balance his book, he charges bettors a lower margin (commonly referred to juice, vig, cut, take, among others), but when the current time $t$ approaches the finishing time $T$, the bookmaker charges a very high margin trying to stay in profit. 
Recall from \eqref{eq:para_exp} that $p_1 = \hat{p} = 0.6$ and $p_2 = 0.4$,  Figure \ref{fig:exp_beta} shows that at the far end $T-t=10$, the bookmaker charges a margin of about 10\% but aggressively increases to 30\% or more when $t$ gets close to $T$. We also observe that the scale of increase is more significant on outcome $A_2$ than on outcome $A_1$, which is because outcomes of small probabilities impose higher risk to the bookmaker. 

In this section, the bookmaker is risk averse equipped with an exponential utility function $U(y) = - \ee^{-\gam y}$, where $\gam >0$ is the risk aversion parameter. In the next study, we investigate how the bookmaker's risk aversion $\gamma$ affects the optimal strategy $u^*$ and the value function $V$.  
Since both $u^*$ and $V$ also depend on $\gamma$ through the function $H$ or $G$, an analytical sensitivity result is not available. We then resort to numerical analysis and  plot both $u^*$ and $V$ as a function of $\gamma$ in Figure \ref{fig:exp_gamma}. 
With all other model parameters fixed, we observe that the bookmaker's value function increases as $\gamma$ increases, i.e., a bookmaker with higher risk aversion achieves higher expected utility when following the optimal strategy. 
As the risk aversion parameter $\gamma$ increases, the bookmaker becomes more risk averse and sets the optimal prices on both outcomes at a higher level.

\begin{figure}[h]
\centering
\includegraphics[trim = 1cm 0cm 0.5cm 1cm, clip=true, width = 0.5 \textwidth]{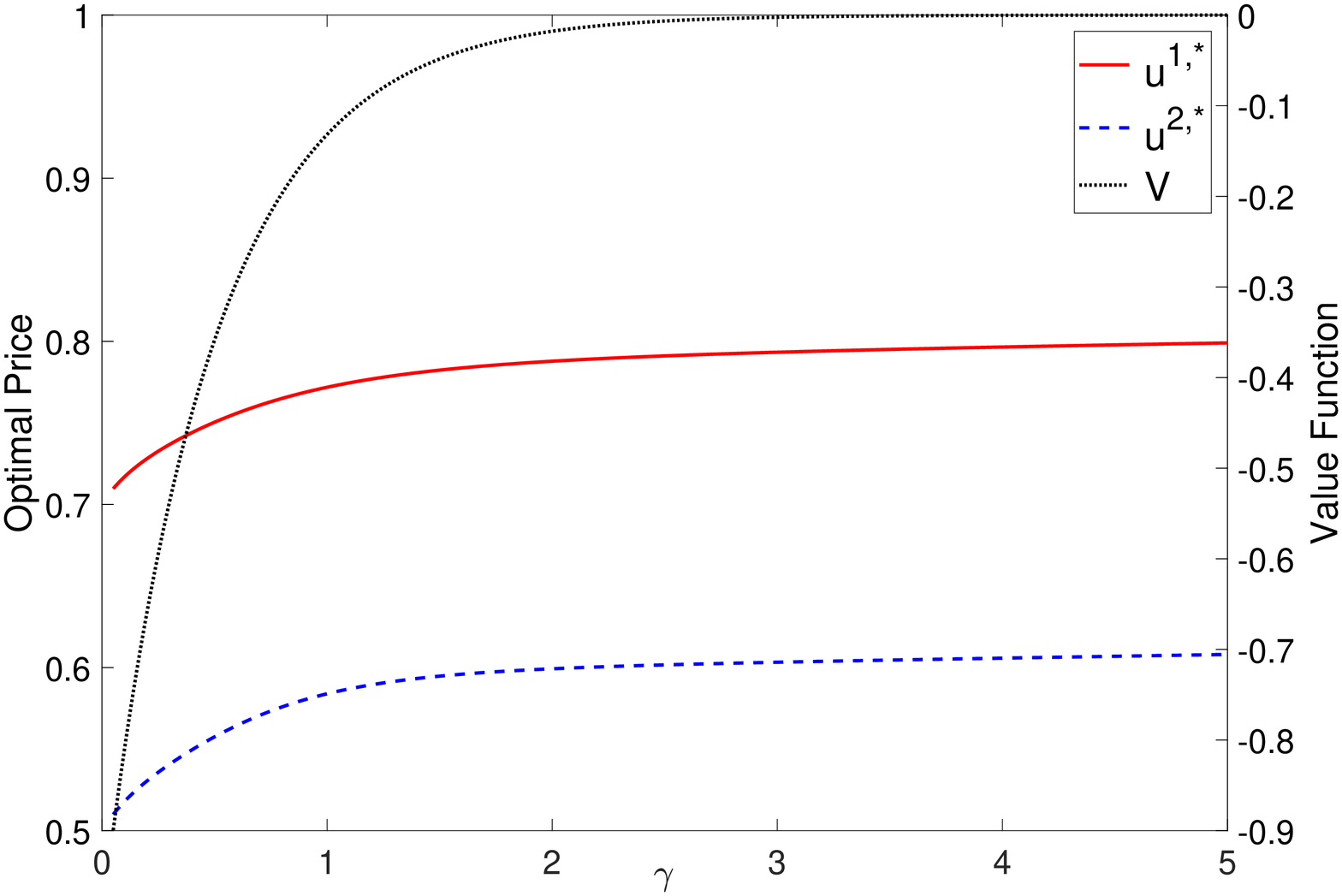}
\\[-2ex]
\caption{\small The graph plots the impact of the risk aversion parameter $\gamma$ on the optimal price strategies $u^{1,*}$ (red solid line) and  $u^{2,*}$ (blue dashed line), and  the value function $V$ (black dotted line). The left $y$-axis is for the optimal price $u^*$ and the right $y$-axis is for the value function $V$. 
	All the parameters other than $\gamma$ are set as in \eqref{eq:para_exp} and $x=1$.}
\label{fig:exp_gamma}
\end{figure}

We also conduct sensitivity analysis on the probability $\hat p$ (the probability of ``Heads'') and the inventory $q_1$ (the number of bets on ``Heads''), with results plotted in Figure \ref{fig:exp_p}. 
The main findings are summarized as follows:
(i) When $\hat p = p_1$ increases, the bookmaker increases the price $u^{1,*}$ on outcome $A_1$ and simultaneously decreases the price $u^{2,*}$ on outcome $A_2$. Notice that when $p_1 = p_2 = 0.5$, the optimal prices on two outcomes are the same.
(ii) Keeping the inventory $q_2$ on outcome $A_2$ constant (we set $q_2 = 5$), the optimal price $u^{1,*}$ on outcome $A_1$ is an increasing function of its inventory $q_1$ while the optimal price $u^{2,*}$ on outcome $A_2$ is a decreasing function of $A_1$'s inventory $q_1$, which is consistent with the result in the left panel of Figure \ref{fig:cor_quan}.

\begin{figure}[h]
	\centering
	\includegraphics[trim = 1cm 0cm 0.5cm 1cm, clip=true,  width = 0.48 \textwidth, height=2in]{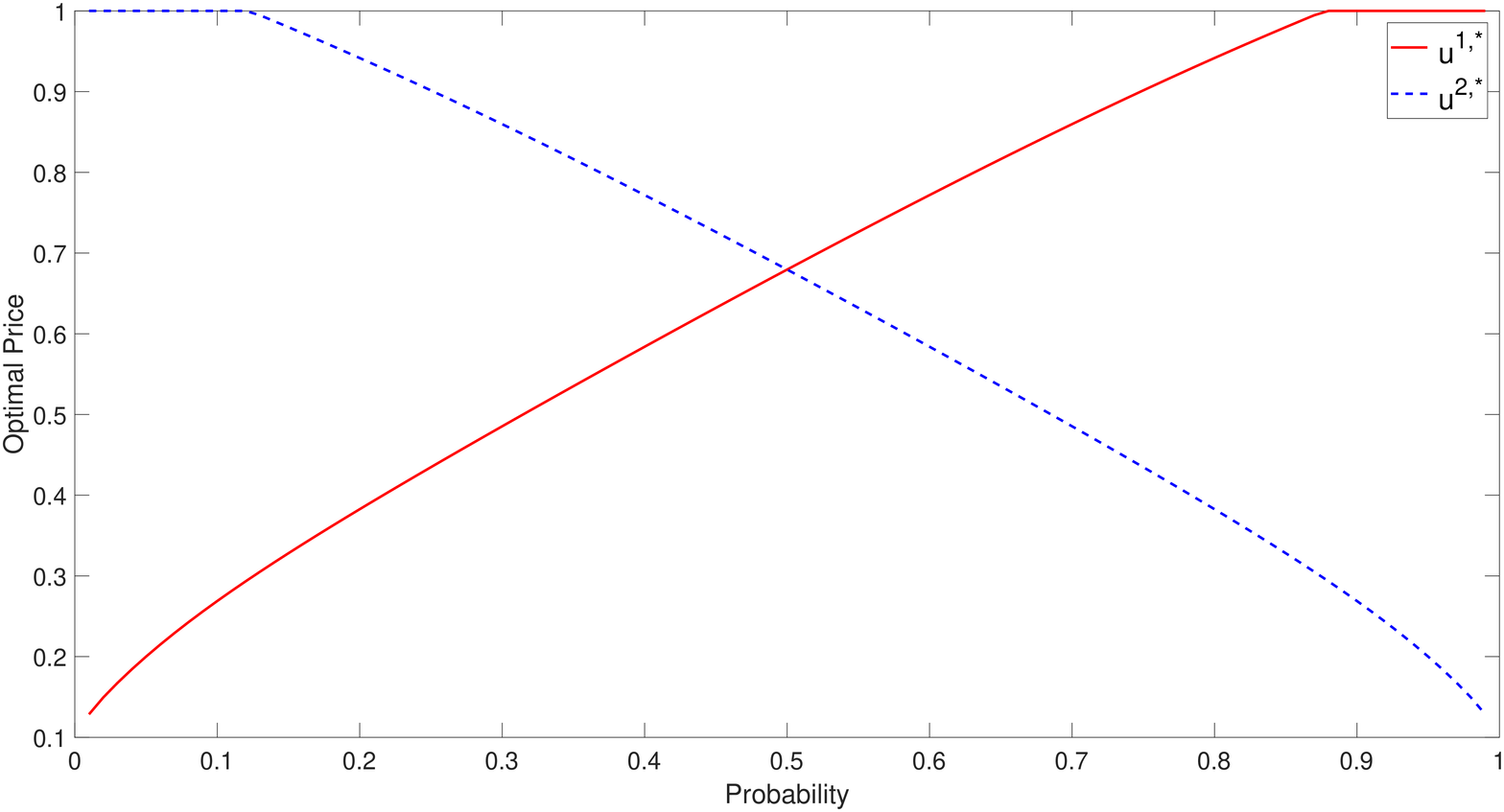}
		\includegraphics[trim = 1cm 0cm 0.5cm 1cm, clip=true,  width = 0.48 \textwidth, height= 2in]{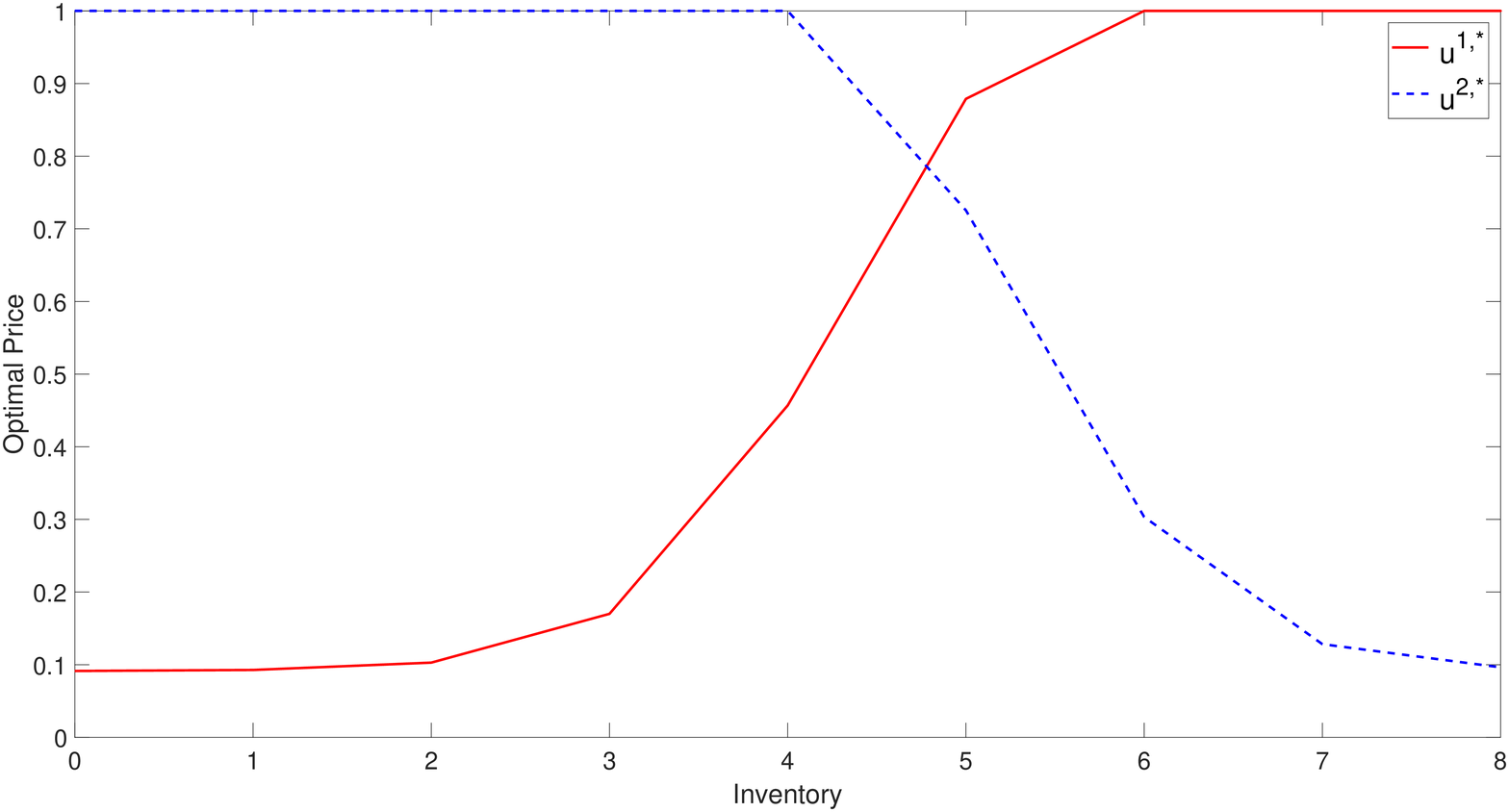}
	\\[-2ex]
	\caption{\small The graph plots the impact of the probability $\hat{p}$ (the probability of ``Heads'') in the left panel and the impact of the inventory $q_1$ (the number of bets on ``Heads'') in the right panel, on the optimal price strategies $u^{1,*}$ (red solid line) and  $u^{2,*}$ (blue dashed line). 
		All the parameters are set as in \eqref{eq:para_exp}, except $\hat{p}$ as a variable in the left panel and $q_1$ as a variable and $q_2 = 5$ in the right panel.}
	\label{fig:exp_p}
\end{figure}

\section{Conclusion}
\label{sec:con}

In this paper, we introduce a general framework for continuous-time betting markets. 
A bookmaker takes bets on random (sporting) events and pays off the winning bets at the terminal time $T$. 
The conditional probability of a set of outcomes is exogenous and may be stochastic, and the bets placed on a set of outcomes may arrive either at a continuous rate or as a state-dependent Poisson process.
The bookmaker controls (updates) the prices of bets dynamically.  In turn, the prices set by the bookmaker affect the rates or intensities of bet arrivals. 
The bookmaker seeks to set prices in order to maximize his expected terminal wealth or expected utility of terminal wealth. 
We are able to obtain either explicit solutions or characterizations of such optimal bookmaking problems in a number of distinct settings.

To solve the optimal bookmaking problem \eqref{prime_prob}, we consider two settings for the bookmaker's utility function $U$: the risk neutral setting and the exponential utility (risk averse) setting.  We end this paper by commenting on the challenges when the bookmaker is risk averse with a power utility function.\footnote{We are indebted to an anonymous referee for motivating us discuss these challenges faced in the power utility setting.} Let us illustrate the main difficulty under the power utility setting of $U(y) = \frac{1}{\alpha} \, y^\alpha$, where $\alpha < 1$ and $\alpha \neq 0$, via a simple example when Assumption \ref{assumption:con_de} holds (i.e., under the continuous arrival model \eqref{eq:det_arrival} and the constant conditional probabilities). Suppose $(A_i)_{i=1,2,\cdots,n}$ form a partition of the sample space with $\Pb(A_i) = p_i  \in (0,1)$, and the rate function $\lam$ is given by \eqref{eq:lambda-1}.
	Let us further restrict  to time-independent constant strategies only $u = (u_1, u_2, \cdots, u_n) \in [0,1]^n$.
	For any strategy $u$, the bookmaker's final wealth $Y_T^u$ is given by 
$Y_T^u = x -  \sum_{i=1}^n \, q_i \Ib_{A_i} + (T-t) \left( \sum_{i=1}^n \, \frac{p_i (1-u_i)}{1-p_i}  -
		\sum_{i=1}^n \,  \frac{p_i (1-u_i)}{u_i (1-p_i)} \Ib_{A_i} \right).$
	Maximizing $\Eb_{t,x,p,q}[U(Y_T^u)]$ leads to the following  system of $n$ non-linear equations:
$$(1- u_i^2)/u_i^2 \cdot \big( Y_i^u \big)^{\alpha -1}  p_i = \sum_{j \neq i} \, \big( Y_j^u \big)^{\alpha -1} p_j ,$$
	where $Y_i^u $ is the bookmaker's profit when $A_i$ occurs at $T$ (i.e., $Y_i^u = Y_T^u(\omega)$, where $\omega \in A_i$).
	Notice that $Y_i^u$ depends on the state variables and controls of \emph{all} outcomes, i.e., $Y_i^u = Y_i^u(t,x,p,q_i,u)$.
Hence, the above non-linear system is fully coupled, while in comparison, the first-order condition in the exponential utility case leads to a decoupled system and $u_i^*$ can be explicitly obtained; see \eqref{e4} and \eqref{e9}. 
The exact same issue appears when we analyze the HJB equation in more complicated cases, which prevents us from guessing an ansatz in a separable form. 
On the numerical side, the corresponding HJB equation is at least (5+1)-dimensional $(t,x,p_1,p_2,q_1,q_2)$. Solving an HJB equation in such high-dimension with any sort of accuracy is extremely difficult due to the curse of dimensionality. In addition, even if the HJB equation could be solved numerically, it is really the optimal control rather than the value function that is sought after.  And deriving the optimal control from a value function that has been solved numerically is even more challenging than solving the HJB equation numerically because the control is written in terms of derivatives of the value function.

\section*{Acknowledgments}
The authors would like to thank two anonymous referees, Editor Steffen Rebennack,  Agostino Capponi, Xuedong He, and Steven Kou for valuable comments and suggestions. We also thank seminar participants at Rutgers University and the University of Michigan, and the organizers (Kim Weston, Kasper Larsen, Erhan Bayraktar, and Asaf Cohen) for their hospitality. Bin Zou acknowledges a start-up grant from the University of Connecticut.

\bibliographystyle{apalike}
\bibliography{reference}

\clearpage
\setcounter{page}{1}
\pagenumbering{roman}

\appendix

\begin{center}
\Large Online Appendix for ``Optimal Bookmaking'' \\
by Matthew Lorig, Zhou Zhou and Bin Zou
\end{center}

\section{Technical Proofs and Results in Section \ref{sec:exp}}

\begin{proof}[Proof to Lemma \ref{lem:sol_hat_V4}]
	Denote by $b:=\max\{b_1,\dotso,b_n \}$. For $k$ large enough, we have
	\begin{align}
	|\alpha_k(q)|\leq\frac{1}{k!}\sum_{\j\in\N^n:|\j|=k} b^k=\frac{1}{k!}{k+n-1\choose n-1}b^k\sim \frac{1}{k!}\frac{k^{n-1}}{(n-1)!}b^k<\frac{k^n}{k!}b^k\sim\frac{b^{k-n}}{(k-n)!}b^n,
	\end{align}
	which proves the first result.
	To show the second result, we obtain
	\begin{align}
	\d_t G(t,q) &=-\sum_{k=1}^\infty k \cdot  \alpha_k(q) \, (T-t)^{k-1}=-\sum_{k=0}^\infty(k+1) \, \alpha_{k+1}(q) \cdot (T-t)^k, \\
	\sum_{i=1}^n h_i G(t,q+\e_i) &=\sum_{k=0}^\infty\left(\sum_{i=1}^n h_i \cdot \alpha_k(q+\e_i)\right)(T-t)^k.
	\end{align}
	Since we have
	\begin{align*}
	\sum_{i=1}^n h_i\cdot \alpha_k(q+\e_i)&=\sum_{i=1}^n h_i \, \frac{1}{k!}\sum_{\j\in\N^n:|\j|=k} \left(h_1^{j_1}\dotso h_n^{j_n}\right) \cdot  d(q+\e_i+\j)\\
	&=\frac{1}{k!}\sum_{i=1}^n\sum_{\j\in\N^n:|\j|=k} h_i\cdot \left(h_1^{j_1}\dotso h_n^{j_n}\right) \cdot  d(q+\e_i+\j)\\
	&=\frac{1}{k!}\sum_{\j'\in\N^n:|\j'|=k+1} \left(h_1^{j_1'}\dotso h_n^{j_n'}\right) \cdot d(q+\j')\\
	&=(k+1)\alpha_{k+1}(q),
	\end{align*}
	the second result follows.
\end{proof}

In Section \ref{sec:exp}, we do not impose any upper bound on the number of bets the bookmaker takes. 
However, if the bookmaker sets an upper for each betting event, we need to modify the results in Theorem \ref{thm:exp_dyna_v1} as described in the corollary below.

\begin{coro-non}
	Let Assumption \ref{assumption:exp} hold. Assume the total number of bets placed on set of outcomes $A_i$ is at most $m_i$, where $i \in \Nb$. 
	The value function to Problem \eqref{prime_prob} is given by 
	\begin{align}
	V(t,x,p,q) 
	= - \ee^{-\gam x} \, \left[\check{G}(t,q) \right]^{- 1/c},
	\end{align}
	where $\check{G}$ is defined by 
	\begin{align}
	\check{G}(T,q) &=  d(q), &
	\check{G}(t,q) &= \sum_{k=0}^{|m| - |q|} \check\alpha_k(q) \cdot (T-t)^k, \qquad t \in [0,T),
	\end{align}
	with functions $\check\alpha_k(q) $ given by 
	\begin{align}
	\check\alpha_0(q) := d(q), \qquad  
	\check \alpha_k(q) :=\frac{1}{k!}\sum_{\j\in I(k,q)} \left(\prod_{i=1}^{n} h_i^{j_i}\right) \cdot d(q+\j),
	\end{align}
	for all $k=1,2,\cdots, \sum_{i=1}^{n} m_i$ and 
	\begin{align}
	I(k,q):= \{\j\in\N^n:\ |\j|=k, \, q+\j\leq m:=(m_1,m_2,\cdots,m_n)\}.
	\end{align}
	The optimal price process $u^*=(u_s^*)_{s \in [t,T]}$ to Problem \eqref{prime_prob} is given by 
	\begin{align}
	u_s^{i,*} = - \frac{1}{\gam}\log\left[\dfrac{\beta \cdot \check H(s,Q_s^*) }{(\beta+\gamma)
		\cdot	\check H(s,Q_s^*+ \e_i) } \right], 
	\end{align}
	for all $i \in I(q)$, where $\check H (t,q) :=  [\check G(t,q)]^{- 1/c}$ and
	\begin{align}
	\label{eq:Iq}
	I(q) := \{i:\ q_i<m_i\} \subset \Nb = \{1,\dotso,n\}.
	\end{align}
\end{coro-non}

\begin{proof}
	With the extra upper bound assumption, the bets on $A_i$ will arrive according to a Poisson process at intensity $\lambda_i(u_i)$ if the total number is less than $m_i$; and 0 if otherwise. 
	Notice that all the equations \eqref{e7}, \eqref{e4}, \eqref{e8}, and \eqref{e5} still hold, except that the summation over index $i$ in these equations will be restricted to the set $I(q)$, defined by \eqref{eq:Iq}. All the results follow naturally by Theorem \ref{thm:exp_dyna_v1}.
\end{proof}

\end{document}